\documentclass[letter]{article}
\usepackage{array}
\usepackage[sectionbib]{chapterbib}

\usepackage{amsmath}
\usepackage{mathtools}
\usepackage{amssymb}
\usepackage[mathscr]{euscript} 

\usepackage{algorithm}
\usepackage{algorithmic}

\usepackage{enumerate}
\usepackage{enumitem}

\usepackage{booktabs} 

\usepackage{fullpage}

\usepackage{etoolbox} 

\usepackage{amsthm} 

\usepackage{enumerate}

\usepackage{ragged2e}
\newcolumntype{P}[1]{>{\RaggedRight\hspace{0pt}}p{#1}}
\newcolumntype{C}[1]{>{\centering\hspace{0pt}}p{#1}}
\usepackage{multirow}
\usepackage{multicol}

\usepackage{graphicx}
\usepackage{epstopdf}
\graphicspath{{Figures/}}

\newtheorem{theorem}{Theorem}

\newtheorem{corollary}[theorem]{Corollary}

\newtheorem{definition}{Definition}

\newcounter{myoptimizationproblemctr}
\newenvironment{myoptimizationproblem}{
   \bigskip\noindent
   \refstepcounter{myoptimizationproblemctr}
   $(\mathbf{P\themyoptimizationproblemctr})$ 
   }{}   

\begin{document}





\title{Distributed Resource Allocation in 5G Cellular Networks\footnote{Book chapter in \emph{Towards 5G: Applications, Requirements and Candidate Technologies},  Wiley, 2015, (Eds. Rath Vannithamby and Shilpa Telwar).
}}
\author{Monowar Hasan and Ekram Hossain \\
University of Manitoba, Canada}
\date{}
\maketitle

\section{Introduction}

The fifth generation (5G) cellular networks are expected to provide wide variety of high rate (i.e., 300 Mbps and 60 Mbps in downlink and uplink,
respectively, in 95 percent of locations and time \cite{metis_5g})  multimedia services. The 5G communication platform is seen as a global unified standard with seamless connectivity among existing standards, e.g., High Speed Packet Access (HSPA), Long Term Evolution-Advanced (LTE-A) and Wireless Fidelity (WiFi). Some of the emerging features and trends of 5G networks are: multi-tier dense heterogeneous networks \cite{horizon_5g, toshiba_5g}, device-to-device (D2D) and machine-to-machine (M2M) communications \cite{toshiba_5g, d2d_5g}, densification of the heterogeneous base stations (e.g., extensive use of relays and small cells) \cite{nw_dens_5g}, cloud-based radio access network \cite{toshiba_5g}, integrated use of  multiple  radio  access  technologies \cite{multi_rat_5g}, wireless network virtualization \cite{toshiba_5g},  massive and 3D MIMO \cite{toshiba_5g, mimo_5g}, millimeter wave \cite{mmw_5g} and full duplex \cite{5g_shilpa} communications.

The 5G cellular wireless systems will have a multi-tier architecture consisting  of  macrocells,  different  types  of  licensed  small  cells and  D2D networks to serve users with different quality-of-service (QoS) requirements in a spectrum efficient manner. Distributed resource allocation and interference management is one of the fundamental research challenges for such multi-tier heterogeneous networks. In this chapter, we consider the radio resource allocation problem in a multi-tier orthogonal frequency division multiple access (OFDMA)-based cellular (e.g., 5G LTE-A) network. In particular, we present three novel approaches for distributed resource allocation  in such networks utilizing the concepts of stable matching, factor-graph based message passing, and distributed auction. 

Matching theory, a sub-field of economics, is a promising concept for distributed resource management in wireless networks. The matching theory allows low-complexity algorithmic manipulations to provide a decentralized self-organizing solution to the resource allocation problems. In matching-based resource allocation, each of the agents (e.g., radio resources and transmitter nodes) ranks the opposite  set using a preference relation. The solution of the matching is able to assign the resources with the transmitters depending on the preferences.

The message passing approach for resource allocation provides low (e.g., polynomial time) complexity solution by distributing the computational load among the nodes in the network.  In the radio resource allocation problems, the decision making agents (e.g., radio resources and the transmitters) form a virtual graphical structure. Each node computes and exchanges simple messages with neighboring nodes in order to find the solution of the resource allocation problem.

Similar to matching based allocation, auction method is also inherited from economics and used in wireless resource allocation problems. Resource allocation algorithms based on auction method provides polynomial complexity solution which are shown to output near-optimal performance. The auction process evolves with a bidding process, in which unassigned agents (e.g., transmitters) raise the cost and bid for resources simultaneously. Once the bids from all the agents are available, the resources are assigned to the highest bidder.


We illustrate each of the modeling schemes with respect to a practical radio resource allocation problem. In particular, we consider a multi-tier network consisting a macro base station (MBS), a set of small cell base
stations (SBSs) and corresponding small cell user equipments (SUEs), as well as D2D user equipments (DUEs). There is a common set of radio resources (e.g., resource blocks [RBs]) available to the network tiers (e.g., MBS, SBSs
and DUEs). The SUEs and DUEs use the available resources (e.g., RB and power level) in an underlay manner as long as the interference caused to the macro tier (e.g., macro user equipments [MUEs]) remains below a given threshold. The goal of resource allocation is to allocate the available RBs and transmit power levels to the SUEs and DUEs in order to maximize the spectral efficiency without causing significant interference to the MUEs. We show that due to the nature of the resource allocation problem, the centralize solution is computationally expensive and also incurs huge signaling overhead. Therefore, it may not be feasible to solve the problem by a single centralized controller node (e.g., MBS) especially in a dense network. Hence distributed solutions with low signaling overhead is desirable.


We assume that readers are familiar with the basics of OFDMA-based cellular wireless networks (e.g., LTE-A networks), as well as have preliminary background on theory of computing (e.g., data structures, algorithms and computational complexity). Followed by a brief theoretical overview of the modeling tools (e.g., stable matching, message passing
and auction algorithm), we present the distributed solution approaches for the resource allocation problem in the aforementioned network setup. We also provide a brief qualitative comparison in terms of various performance 
metrics such as complexity, convergence, algorithm overhead etc. 

The organization of the rest of the chapter is as follows: the system model, related assumptions, and the resource allocation problem is presented in Section \ref{sec:sys_model}. The disturbed solutions for resource allocation problem, e.g., stable matching, message passing and auction method are discussed in the Sections \ref{sec:sm_ra}, \ref{sec:mp_ra}, \ref{sec:am_ra}, respectively. The qualitative comparisons among the resource allocation approaches are presented in Section \ref{sec:comparisons}. We conclude the chapter in Section \ref{sec:conclusion} highlighting the directions for future research. Key mathematical symbols and notations used in the chapter are summarized in Table \ref{tab:notations}.

\begin{table}[!h]
\centering
\begin{footnotesize}
\begin{tabular}{c P{10.0cm}}
\toprule
\multicolumn{1}{c}{Notation} & \multicolumn{1}{c}{Physical Interpretation} \\
\midrule
\multicolumn{2}{l}{$\bullet$~\textit{Network model:}}  \\
$\mathcal{U}^{\mathrm m}$, $\mathcal{U}^{\mathrm s}$, $\mathcal{U}^{\mathrm d}$ & Set of MUE, SUE and D2D pairs, respectively \\
$\mathcal{K}^{\mathrm T}$, $\mathcal{K}^{\mathrm R}$ & Set of underlay transmitters and receivers, respectively  \\
$\mathcal{N}$, $\mathcal{L}$ & Set of RBs and power levels, respectively \\
$K$, $N$, $L$ & Total number of underlay transmitters, RBs, and power levels, respectively \\
$u_k$ & The UE associated with underlay transmitter $k$ \\
$x_{k}^{(n,l)}, \mathbf{X}$ & Allocation indicator, whether transmitter $k$ using resource $\lbrace n, l \rbrace$ and the indicator vector, respectively\\
$g_{i,j}^{(n)}$ & Channel gain between link $i,j$ over RB $n$ \\
$\gamma_{u_k}^{(n)}$ & SINR in RB $n$ for the UE $u_k$\\
$\Gamma_{u_k}^{(n,l)}$ & Achievable SINR of the UE $u_k$ over RB $n$ using power level $l$\\
$p_{k}^{(n)}$ & Transmit power of transmitter $k$ over RB $n$\\
$R_{u_k}$ & Achievable data rate for $u_k$ \\
$I^{(n)}$, $I_{\mathrm{max}}^{(n)}$ & Aggregated interference and threshold limit for the RB $n$, respectively \\
$\mathfrak{U}_{k}^{(n,l)}$ & Utility for transmitter $k$ using resource $\lbrace n, l \rbrace$ \\
\midrule
\multicolumn{2}{l}{$\bullet$~\textit{Stable matching:}}  \\
$\mu$ & Matching (e.g., allocation) of transmitter to the resources \\
$i_1 \succeq_j  i_2$ & Preference relation for agent $j$ (i.e., $i_1$ is more preferred than $i_2$) \\
$\boldsymbol{\mathscr{P}}_{k}(\mathcal{N}, \mathcal{L})$, $\boldsymbol{\mathscr{P}}_{n}(\mathcal{K}^{\mathrm T}, \mathcal{L})$ & Preference profile for the transmitter $k$ and RB $n$, respectively \\
\midrule 
\multicolumn{2}{l}{$\bullet$~\textit{Message passing:}}  \\
$\delta_{\lbrace n,l \rbrace \rightarrow k} \big( x_{k}^{(n,l)} \big)$ & Message delivered by the resource $\lbrace n,l \rbrace$ to the transmitter $k$\\
$\delta_{k \rightarrow \lbrace n,l \rbrace} \big( x_{k}^{(n,l)} \big)$ & Message from transmitter $k$ to the resource $\lbrace n,l \rbrace$  \\
$\psi_{\lbrace n,l \rbrace \rightarrow k}$ & Normalized message from the resource $\lbrace n,l \rbrace$ to the transmitter $k$ \\
$\psi_{ k \rightarrow \lbrace n,l \rbrace }$ & Normalized message from the transmitter $k$ to the resource $\lbrace n,l \rbrace$\\
$\tau_{k}^{(n,l)}$ & Node marginals for the transmitter $k$ using resource $\lbrace n,l \rbrace$ \\ 
\midrule
\multicolumn{2}{l}{$\bullet$~\textit{Auction method:}}  \\
$C_{k}^{(n,l)}$ & Cost for transmitter $k$ using resource $\lbrace n,l \rbrace$ \\ 
$B_{k}^{(n,l)}$ & Data rate (multiplied by a weighting factor) achieved by transmitter $k$ using resource $\lbrace n,l \rbrace$ \\ 
$\mathfrak{b}_{k}^{( n,l)}$ & Local bidding information available to transmitter $k$ for the resource $\lbrace n,l \rbrace$ \\
$\epsilon$ & Minimum bid increment parameter \\
$\Theta_{k} = \lbrace n,l \rbrace$ & Assignment of resource $\lbrace n,l \rbrace$ to the transmitter $k$\\
\midrule
\multicolumn{2}{l}{$\bullet$~\textit{Miscellaneous:}}  \\
$|\mathbf{y}| $ & Length of the vector $\mathbf{y}$ \\
$y(t)$ & Value of variable $y$ at any iteration $t$ \\
$z := y$ & Assignment of the value of variable $y$ to the variable $z$ \\
\textit{/* comment */} & Commented text inside algorithms \\
\toprule
\end{tabular}
\end{footnotesize}
\caption{List of major notations}{}
\label{tab:notations}
\end{table}

\section{System Model} \label{sec:sys_model}

\subsection{Network Model and Assumptions} \label{subsec:nw_model}

\begin{figure}[h t b]
\centering
\includegraphics[width=3.0in]{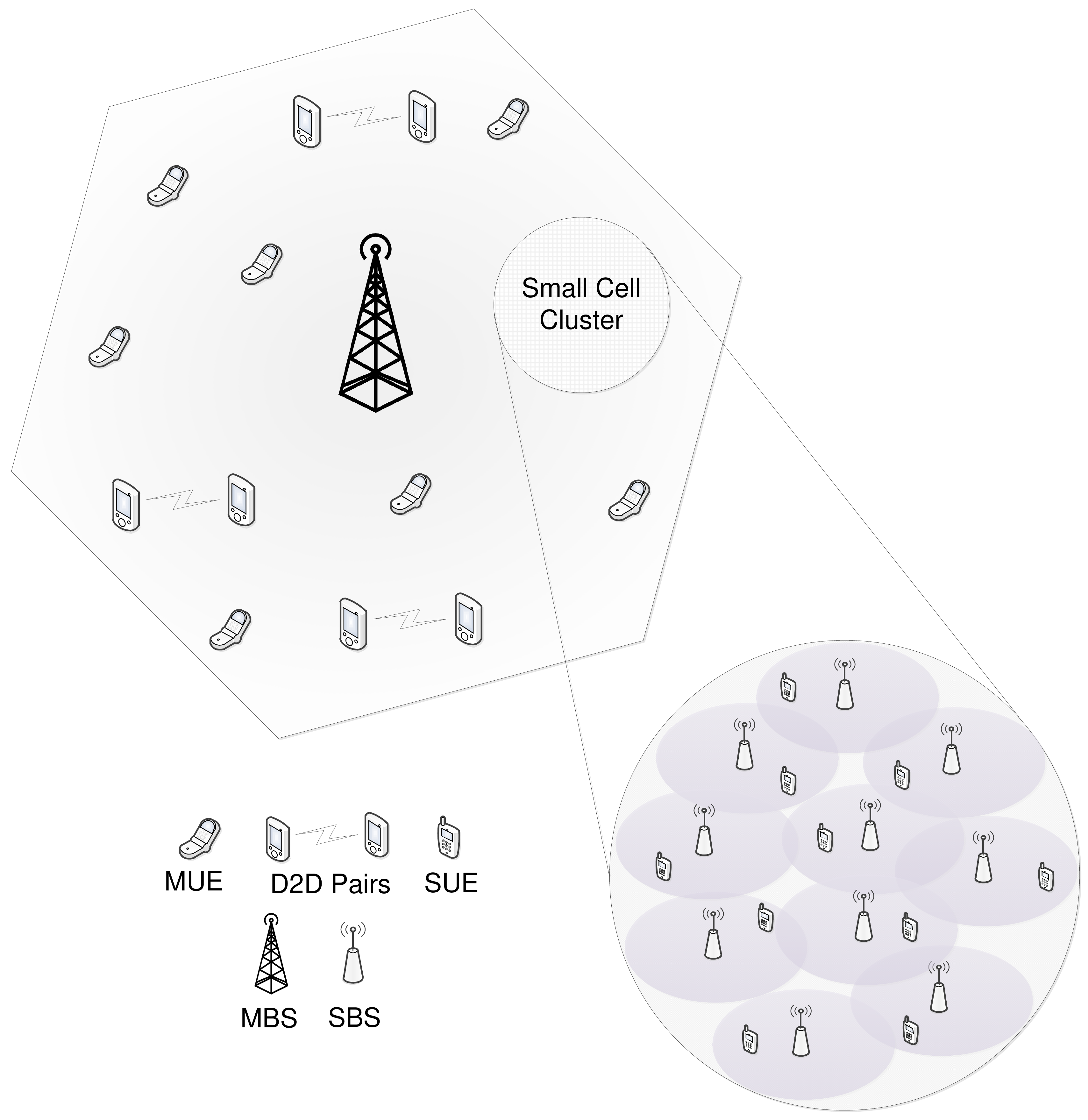}
 \caption{Schematic diagram of the heterogeneous network model. The D2D pairs, SBSs and SUEs are underlaid within the macro tier by reusing same set of radio resources.}
\label{fig:sys_mod}
\end{figure}

Let us consider a transmission scenario of heterogeneous network as shown in Fig. \ref{fig:sys_mod}. The network consists of one MBS and a set of $C$ cellular MUEs, i.e., $\mathcal{U}^{\mathrm m} = \lbrace 1,2,\cdots, C \rbrace$. There are also $D$ D2D pairs and a cluster of $S$ SBSs located within the coverage area of the MBS. The set of SBSs is denoted by $\mathcal{S} = \lbrace 1, 2, \cdots S\rbrace$. For simplicity we assume that each SBS serves only one SUE for a single time instance and the set of SUE is given by $\mathcal{U}^{\mathrm s} = \lbrace 1,2,\cdots, S \rbrace$. The set of D2D pairs is denoted as $\mathcal{U}^{\mathrm d} = \lbrace 1,2,\cdots, D \rbrace$. In addition, the $d$-th element of the sets $\mathcal{U}^{\mathrm d_{T}}$ and $\mathcal{U}^{\mathrm d_{R}}$ denotes the transmitter and receiver UE of the D2D pair $d \in \mathcal{U}^{\mathrm d}$, respectively. The set of UEs in the network is given by $\mathcal{U} = \mathcal{U}^{\mathrm m} \cup \mathcal{U}^{\mathrm s} \cup \mathcal{U}^{\mathrm d}$. For notational convenience, we denote by $\mathcal{K}^{\mathrm T} = \mathcal{S} \cup \mathcal{U}^{\mathrm d_T}$ the set of underlay transmitters (e.g., SBSs and transmitting D2D UEs) and $\mathcal{K}^{\mathrm R} =  \mathcal{U}^{\mathrm s} \cup \mathcal{U}^{\mathrm d_{R}}$ denotes the set of underlay receivers (e.g., SUEs and receiving D2D UEs). 


The SBSs and DUEs are underlaid within the \textit{macro tier} (e.g., MBS and MUEs). Both the macro tier and the \textit{underlay tier} (e.g., SBSs, SUEs and D2D pairs) use the same set $\mathcal{N} = \lbrace 1, 2, \cdots N \rbrace$ of orthogonal RBs\footnote{The minimum scheduling unit of LTE-A standard is referred to as an RB. One RB consists of 12 subcarriers (e.g., 180 kHz) in the frequency domain and one sub-feame (e.g., 1 millisecond) in the time domain. For a brief overview of heterogeneous network in the context of LTE-A standard refer to \cite[Chapter 1]{hetnet_book_sir}.}. Each transmitter node in the underlay tier (e.g., SBS and D2D transmitter) selects one RB from the available $N$ RBs. In addition, the underlay transmitters are capable of selecting the transmit power from a finite set of power levels, i.e., $\mathcal{L} = \lbrace 1, 2, \cdots L \rbrace$.  
Each SBS and D2D transmitter should select a suitable RB-power level combination. This RB-power level combination is referred to as \textit{transmission alignment}\footnote{Throughout this chapter we use the term \textit{resource} and \textit{transmission alignment} interchangeably.} \cite{prabo_journal}. For each RB $n \in \mathcal{N}$, there is a predefined threshold $I_{\mathrm{max}}^{(n)}$ for maximum aggregated interference caused by the underlay tier to the macro tier. We assume that value of $I_{\mathrm{max}}^{(n)}$ is known to the underlay transmitters by using the feedback control channels. An underlay transmitter (i.e., SBS or transmitter DUE) is allowed to use the particular transmission alignment as long as the cross-tier interference to the MUEs is within the threshold limit.

The system model  considered here is a \textit{multi-tier heterogeneous network} since each of the network tiers (e.g., macro tier and underlay tier consisting with small cells and D2D UEs) has different transmit power range, coverage region and specific set of users with different application requirements. It is assumed that the user association to the base stations (either MBS or SBSs) is completed prior to resource allocation. In addition, the potential DUEs are discovered during the D2D session setup by transmitting known synchronization or reference signal (i.e., beacons) \cite{network_asst_d2d}. According to our system model, only one MUE is served on each RB to avoid co-tier interference within the macro tier. However multiple underlay UEs (e.g., SUEs and DUEs) can reuse the same RB to improve the spectrum utilization. This reuse causes severe cross-tier interference to the MUEs, and also co-tier interference within the underlay tier; which leads the requirement of an efficient resource allocation scheme. 

\subsection{Achievable Data Rate}

The MBS transmits to the MUEs using a fixed power $p_{M}^{(n)} > 0$ for $\forall n$.  For each underlay transmitter $k \in \mathcal{K}^{\mathrm T}$, the transmit power over the RBs is determined by the vector $\mathbf{P}_{\mathrm k} = \left[ p_k^{(1)}, p_k^{(2)}, \cdots, p_k^{(N)} \right]^{\mathsf{T}}$ where $p_k^{(n)} \geq 0$ denotes the the transmit power level of the transmitter $k$ over RB $n$. The transmit power $p_k^{(n)}, ~\forall n$  must be selected from the finite set of power levels $\mathcal{L}$. Note that if the RB $n$ is not allocated to the transmitter $k$, the corresponding power variable $p_k^{(n)} = 0$. Since we assume that each underlay transmitter selects only one RB, only one element in the power vector $\mathbf{P}_{\mathrm k}$ is non-zero.
 
 All links are assumed to experience independent block fading. We denote by $g_{i, j}^{(n)}$ the channel gain between the links $i$ and $j$ over RB $n$ and defined by $g_{i, j}^{(n)} = \beta_{i,j}^{(n)} d_{i,j}^{-\alpha} $ where  $\beta_{i,j}^{(n)}$ denote the channel fading component between link $i$ and $j$ over RB $n$, $d_{i,j}$ is the distance between node $i$ and $j$, and $\alpha$ is the path-loss exponent.
 
 For the SUEs, we denote $u_k$ as the SUE associated to SBS $k \in \mathcal{S}$, and for the DUEs, $u_k$ refer to the receiving D2D UE of the D2D transmitter $k \in \mathcal{U}^{\mathrm d_{T}}$. The received signal-to-interference-plus-noise ratio (SINR) for the any arbitrary SUE or D2D receiver, i.e., $u_k \in \mathcal{K}^{\mathrm R}, k \in \mathcal{K}^{\mathrm T}$ over RB $n$ is given by
 \begin{equation} \label{eq:sinr_underlay}
 \gamma_{u_k}^{(n)} = \frac{g_{k, u_k}^{(n)}p_{k}^{(n)}}{\underbrace{ g_{M, u_k}^{(n)}p_{M}^{(n)}}_\text{interference from macro tier} + \underbrace{\sum\limits_{\substack{ k^\prime \in  \mathcal{K}^{\mathrm T}, k^\prime \neq k  }} g_{k^\prime, u_k}^{(n)} p_{k^\prime}^{(n)}}_\text{interference from underlay tier} + ~\sigma^2}
 \end{equation}
 where $g_{k,u_k}^{(n)}$ is the link gain between the SBS and SUE (e.g., $u_k \in \mathcal{U}^{\mathrm s},  k \in \mathcal{S}$) or the link gain between the D2D UEs (e.g., $u_k \in \mathcal{U}^{\mathrm d_{R}}, k \in \mathcal{U}^{\mathrm d_T}$), and $g_{M, u_k}^{(n)}$ is the interference gain between the MBS and the UE $u_k$. In Equation (\ref{eq:sinr_underlay}), the variable $\sigma^2 = N_0 B_{\mathrm{RB}}$ where $B_{\mathrm {RB}}$ is the bandwidth corresponding to an RB and $N_0$ denotes the thermal noise. Similarly, the SINR for the MUE $m \in \mathcal{U}^{\mathrm m}$ over RB $n$ can be written as follows:
 \begin{equation}
 \gamma_{m}^{(n)} = \frac{g_{M, m}^{(n)}p_{M}^{(n)}}{\sum\limits_{\substack{ k \in   \mathcal{K}^{\mathrm T} }} g_{k, m}^{(n)} p_{k}^{(n)} + ~\sigma^2}.
 \end{equation}
  Given the SINR, the data rate of the UE $u \in \mathcal{U}$ over RB $n$ can be calculated according to the Shannon's formula, i.e., $R_{u}^{(n)} = B_{\mathrm {RB}} \log_2 \left(1 +  \gamma_{u}^{(n)} \right)$.

\subsection{Formulation of the Resource Allocation Problem} \label{subsec:rap_org}

The objective of resource (i.e., RB and transmit power) allocation problem is to obtain the assignment of RB and power level (e.g., transmission alignment) for the underlay UEs (e.g., D2D UEs and SUEs) that maximizes the achievable sum data rate. The RB and power level allocation indicator for any underlay transmitter $k \in \mathcal{K}^{\mathrm T}$ is denoted by a binary decision variable $x_{k}^{(n, l)}$ where
\begin{equation}
x_{k}^{(n, l)}  = \begin{cases}
 1, \quad  \text{if the transmitter $k$ is trasnmitting over RB $n$ with power level $l$} \\
 0, \quad \text{otherwise.}
\end{cases}
\end{equation}
Note that the decision variable $x_{k}^{(n, l)} = 1$ implies that $p_k^{(n)} = l$. Let $K = S + D$ denote the total number of underlay transmitters. The achievable data rate of an underlay UE $u_k$ with the corresponding transmitter $k$ is written as 
\begin{equation} \label{eq:rate_ue}
R_{u_k} = \sum\limits_{n = 1}^{N} \sum\limits_{l = 1}^{L}  ~x_{k}^{(n,l)}   B_{\mathrm {RB}} \log_2 \left(1 +  \gamma_{u_k}^{(n)} \right). 
\end{equation}
The aggregated interference experienced on RB $n$ is given by $I^{(n)} = \sum\limits_{k =1}^{K}\sum\limits_{l = 1}^{L}x_{k}^{(n, l)} g_{k,m_k^*}^{(n)} p_k^{(n)}$, where $m_k^* = \underset{m}{\operatorname{argmax}}~ g_{k,m}^{(n)}, ~\forall m \in \mathcal{U}^{\mathrm m}$. 
In order to calculate the aggregated interference $I^{(n)}$ on RB $n$ we use the concept of reference user \cite{ref_user}. For any RB $n$, the interference caused by the underlay transmitter $k$ is determined by the highest gains between the transmitter $k$ and MUEs, e.g., the MUE $m_k^*$ who is the mostly affected UE by the transmitter $k$. Satisfying the interference constraints considering the gain with reference user will also satisfy the interference constraints for other MUEs. As mentioned in Section \ref{subsec:nw_model}, an underlay transmitter is allowed to use a particular transmission alignment only when it does not violate the interference threshold to the MUEs, i.e., $I^{(n)} < I_{\mathrm{max}}^{(n)}, ~\forall n$. Mathematically, the resource allocation problem can be expressed by using the following optimization formulation:

\begin{myoptimizationproblem} \label{opt:combopt}
\vspace*{-2.0em}
\begin{subequations}
\begin{align}
\hspace{3em} \underset{x_{k}^{(n,l)},~ p_k^{(n)}}{\operatorname{max}} ~ \sum_{\substack{k =1}}^{K}  \sum_{n = 1}^{N} \sum_{l = 1}^{L}  ~x_{k}^{(n,l)}   & B_{\mathrm {RB}}  \log_2\left( 1 + \gamma_{u_k}^{(n)} \right)  \nonumber \\
 \text{subject~ to:} \hspace{7em} \nonumber\\
\sum_{k =1}^{K}\sum_{l = 1}^{L}x_{k}^{(n, l)} g_{k,m_k^*}^{(n)} p_k^{(n)} &< I_{\mathrm{max}}^{(n)}, \quad \forall n \in\mathcal{N} \label{eq:opt_intf}\\
\sum_{n = 1}^{N} \sum_{l = 1}^{L} x_{k}^{(n,l)} &\leq 1, \quad \quad ~~\forall k \in \mathcal{K}^{\mathrm T} \label{eq:opt_rbpw}\\
x_{k}^{(n,l)} &\in \lbrace 0, 1 \rbrace, ~~~ \forall k \in \mathcal{K}^{\mathrm T},~\forall n \in \mathcal{N},~\forall l \in \mathcal{L}  \label{eq:opt_bin}
\end{align}
\end{subequations}
\end{myoptimizationproblem}
where \vspace*{-0.3em} \begin{equation} \label{eq:sinr_formulation}
\gamma_{u_k}^{(n)} = \frac{g_{k, u_k}^{(n)}p_{k}^{(n)}}{ g_{M, u_k}^{(n)}p_{M}^{(n)} + \sum\limits_{\substack{ k^\prime \in \mathcal{K}^{\mathrm{T}},\\ k^\prime \neq k }}^{K} \sum\limits_{l^\prime = 1}^{L} x_{j}^{(n,l^\prime)} g_{k^\prime, u_k}^{(n)} p_{k^\prime}^{(n)} + ~\sigma^2}.
\end{equation} 

The objective of the resource allocation problem $\mathbf{P\ref{opt:combopt}}$ is to maximize the data rate of the SUEs and DUEs subject to the set of constraints given by Equations (\ref{eq:opt_intf})-(\ref{eq:opt_bin}). With the constraint in Equation (\ref{eq:opt_intf}), the aggregated interference caused to the MUEs by the underlay transmitters on each RB is limited by a predefined threshold. The constraint in Equation (\ref{eq:opt_rbpw}) indicates that the number of RB selected by each underlay transmitter should be at most one and each  transmitter can only select one power level at each RB. The binary indicator variable for transmission alignment selection is represented by the constraint in Equation (\ref{eq:opt_bin}).

\begin{corollary}
The resource allocation problem $\mathbf{P\ref{opt:combopt}}$ is a combinatorial non-convex non-linear optimization problem and the centralized solution of the above problem is strongly NP-hard especially for the large set of  ~$\mathcal{U}$, $\mathcal{N}$, and $\mathcal{L}$. 
\end{corollary}



The complexity to solve the above problem using exhaustive search is of $\mathcal{O}\left( \left(NL \right)^{K} \right)$. As an example, when $N=6, L=3, $ and $K = 3+2 = 5$, the decision set (e.g., search space) contains $1889568$ possible transmission alignments. 
Considering the computational overhead, it not feasible to solve the resource allocation problem by a single central controller (e.g., MBS) in a practical system; and such centralized solution approach requires all the channel state information (CSI) available to the MBS.  

Due to mathematical intractability of solving the above resource allocation problem, in the following we present three distributed heuristic solution approaches, namely, stable matching, factor graph based message passing, and distributed auction-based approaches. The distributed solutions are developed under the assumption that the system is feasible, i.e., given the resources and parameters (e.g., size of the network, interference thresholds etc.), it is possible to obtain an allocation that satisfies all the constraints of the original optimization problem.

\section{Resource Allocation Using Stable Matching} \label{sec:sm_ra}

The resource allocation approach using stable matching involves multiple decision-making agents, i.e., the available radio resources (transmission alignments) and the underlay transmitters; and the solutions (i.e., matching between transmission alignments and transmitters) are produced by individual actions of the agents. The actions, i.e., matching requests and confirmation or rejection are determined by the given \textit{preference profiles}, i.e., the agents hold lists of preferred matches over the opposite set each. The matching outcome yields mutually beneficial assignments between the transmitters and available resources that are individually conducted by such preference lists. In our model, the preference could based on CSI parameters and achievable SINR. \textit{Stability} in matching implies that, with regard to their initial preferences, neither the underlay transmitters nor the MBS (e.g., transmission alignments) have an incentive to alter the allocation.

\subsection{Concept of Matching}

A \textit{matching} (i.e., allocation) is given as an assignment of transmission alignment to the underlay transmitters forming the set $\lbrace k, n, l \rbrace \in \mathcal{K}^{\mathrm T} \times \mathcal{N} \times \mathcal{L}$. According to our system model, each underlay transmitter is assigned to only one RB; however, multiple transmitters can transmit on the same RB to improve spectrum utilization. This scheme corresponds to a \textit{many-to-one} matching in the theory of stable matching. More formally the matching can be defined as follows \cite{sm_def}:

\begin{definition} \label{def:matching}
A matching $\mu$ is defined as a function, i.e.,  $\mu: \mathcal{K}^{\mathrm T} \times \mathcal{N} \times \mathcal{L} \rightarrow \mathcal{K}^{\mathrm T}  \times \mathcal{N} \times \mathcal{L}$ such that
\begin{enumerate} [label={\roman*)}] 
\setlength{\itemsep}{0.5pt}%
    \setlength{\parskip}{0pt}%
\item $\mu(k) \in \mathcal{N} \times \mathcal{L} $ and $|\mu_l(n)| \in \lbrace 0,1 \rbrace$ \quad \mbox{and}
\item $\mu(n) \in \left\lbrace\mathcal{K}^{\mathrm T} \times \mathcal{L}\right\rbrace \cup \lbrace \varnothing \rbrace$ and $|\mu(n)| \in \lbrace 1, 2, \ldots, K \rbrace$
\end{enumerate}
where $\mu(k) = \lbrace n, l\rbrace \Leftrightarrow \mu(n) = \lbrace k, l\rbrace$ for $ \forall k \in \mathcal{K}^{\mathrm T}, \forall  n \in \mathcal{N}, \forall  l \in \mathcal{L},$ and $|\mu(\cdot)|$ denotes the cardinality of matching outcome $\mu(\cdot)$.
\end{definition}


The above \textbf{Definition \ref{def:matching}} implies that  $\mu$ is a one-to-one matching if the input to the function is an underlay transmitter. On the other hand, $\mu$ is a one-to-many function, i.e., $\mu_l(n)$ is not unique if the input to the function is an RB. The interpretation of $\mu(n) = \varnothing $ implies that for some RB $n \in \mathcal{N}$ the corresponding RB is unused by any underlay transmitter under the matching $\mu$. The outcome of the matching determines the RB allocation vector  and corresponding power level, e.g., $\mu \equiv  \mathbf{X} $, where 
\begin{equation} \label{eq:rap_X}
\mathbf{X} = \left[x_{1}^{(1, 1)}, \cdots, x_{1}^{(1, L)}, \cdots, x_{1}^{(N, L)}, \cdots, x_{K}^{(N, L)} \right]^{\mathsf{T}}.
\end{equation}

\subsection{Utility Function and Preference Profile}

Let the parameter $\Gamma_{u_k}^{(n, l)} \triangleq {\gamma_{u_k}^{(n)}}_{ \!\! \vert p_k^{(n)} = l }$ denote the achievable SINR of the UE $u_k$ over RB $n$ using power level $l$ (e.g., $p_k^{(n)} = l$) where $\gamma_{u_k}^{(n)}$ is given by Equation (\ref{eq:sinr_formulation}). We express the data rate as a function of SINR. In particular, let $\mathscr{R}\left(\Gamma_{u_k}^{(n, l)}\right) =  B_{\mathrm {RB}} \log_2 \left(1 +  \Gamma_{u_k}^{(n, l)} \right)$ denote the achievable data rate for the transmitter $k$ over RB $n$ using power level $l$. The utility of an underlay transmitter for a particular transmission alignment is determined by two factors, i.e., the achievable data rate for a given RB power level combination, and an additional cost function that represents the aggregated interference caused to the MUEs on that RB. In particular, the \textit{utility} of the underlay transmitter $k$ for a given RB $n$ and power level $l$ is given by 
\begin{equation} \label{eq:sm_utility}
\mathfrak{U}_{k}^{(n,l)} = w_1 \mathscr{R}\left(\Gamma_{u_k}^{(n, l)}\right)  -  w_2 \left( I^{(n)} - I_{\mathrm{max}}^{(n)} \right) 
\end{equation}
where $w_1$ and $w_2$ are the biasing factors and can be selected based on which network tier (i.e., macro or underlay tier) should be given priority for resource allocation \cite{prabo_journal}. As mentioned earlier each underlay transmitter and RB hold a list of preferred matches. The preference profile of an underlay transmitter $k \in \mathcal{K}^{\mathrm T}$ over the set of available RBs $\mathcal{N}$ and power levels $\mathcal{L}$ is defined as a vector of linear order $\boldsymbol{\mathscr{P}}_{k}(\mathcal{N}, \mathcal{L}) = \left[\mathfrak{U}_{k}^{(n,l)} \right]_{n \in \mathcal{N}, l \in \mathcal{L}}$. We denote by $\lbrace n_1, l_1 \rbrace \succeq_k  \lbrace n_2, l_2 \rbrace$ that the transmitter $k$ prefers the transmission alignment $\lbrace n_1, l_1 \rbrace$ to $\lbrace n_2, l_2 \rbrace$, and consequently, $\mathfrak{U}_{k}^{(n_1,l_1)} > \mathfrak{U}_{k}^{(n_2,l_2)}$. Similarly, the each RB holds the preference over the underlay transmitters and power levels given by  $\boldsymbol{\mathscr{P}}_{n}(\mathcal{K}^{\mathrm T}, \mathcal{L}) = \left[\mathfrak{U}_{k}^{(n,l)} \right]_{k \in \mathcal{K}^{\mathrm T}, l \in \mathcal{L}}$.

\subsection{Algorithm Development}

The matching between transmission alignments to the transmitters is performed in an iterative manner as presented in \textbf{Algorithm \ref{alg:ta_sm}}. While a transmitter is unallocated and has a non-empty preference list, the transmitter is temporarily assigned to its first preference over transmission alignments, e.g., the pair of RB and power level, $\lbrace n,l \rbrace$. If the allocation to the RB $n$ does not violate the tolerable interference limit $I_{\mathrm{max}}^{(n)}$, the allocation will persist. Otherwise, until the aggregated interference on the RB $n$ is below threshold, the worst preferred transmitter(s) from the preference list of RB $n$ will be removed even though it was allocated previously. The process terminates when  no more transmitters are unallocated. Since the iterative process dynamically updates the preference lists, the procedure above ends up with a local stable matching \cite{matching_org_paper}.

\begin{algorithm} [!t]
\AtBeginEnvironment{algorithmic}{\small} 
\caption{Assignment of transmission alignments using stable matching}
\label{alg:ta_sm}
\begin{algorithmic}[1]   
\renewcommand{\algorithmicrequire}{\textbf{Input:}}
\renewcommand{\algorithmicensure}{\textbf{Output:}}
\renewcommand{\algorithmicforall}{\textbf{for each}}
\renewcommand{\algorithmiccomment}[1]{\textit{/* #1 */}}

\REQUIRE The preference profiles $\boldsymbol{\mathscr{P}}_{k}(\mathcal{N}, \mathcal{L})$,~ $\forall k \in \mathcal{K}^{\mathrm T}$ and $\boldsymbol{\mathscr{P}}_{n}(\mathcal{K}^{\mathrm T}, \mathcal{L})$,~ $\forall n \in \mathcal{N}$.

\ENSURE The transmission alignment indicator $\mathbf{X} = \left[x_{1}^{(1, 1)}, \cdots, x_{1}^{(1, L)}, \cdots, x_{1}^{(N, L)}, \cdots, x_{K}^{(N, L)} \right]^{\mathsf{T}}$.

\STATE Initialize $\mathbf{X} := \mathbf{0}$.

\WHILE{ some transmitter $k$ is unassigned \AND $\boldsymbol{\mathscr{P}}_{k}(\mathcal{N}, \mathcal{L})$ is non-empty }
\STATE $\left\lbrace n_{\mathrm{mp}},l_{\mathrm{mp}} \right\rbrace :=$ most preferred RB with power level $l_{\mathrm{mp}}$ from the profile $\boldsymbol{\mathscr{P}}_{k}(\mathcal{N}, \mathcal{L})$. 
\STATE Set $x_{k}^{\left(n_{\mathrm{mp}},l_{\mathrm{mp}} \right)} := 1$.  ~\COMMENT{\footnotesize Temporarily assign the RB and power level to the transmitter $k$}


\STATE $\mathfrak{I}^{(n_\mathrm{mp})} := g_{k,m_k^*}^{(n_\mathrm{mp})} l_{\mathrm{mp}} + \!\!  \sum\limits_{\substack{ k^\prime \in \mathcal{K}^{\mathrm{T}},\\k^\prime \neq k}}\sum\limits_{l^\prime = 1}^{L}x_{k^\prime}^{(n_\mathrm{mp}, l^\prime)} g_{k^\prime,m_{k^\prime}^*}^{(n_\mathrm{mp})} p_{k^\prime}^{(n_\mathrm{mp})}$. ~\COMMENT{\footnotesize Estimate interference of $n_\mathrm{mp}$} 

\IF{$\mathfrak{I}^{(n_\mathrm{mp})} \geq I_{\mathrm{max}}^{(\mathrm{mp})} $}
\REPEAT  

\STATE $ \left\lbrace k_{\mathrm{lp}}, l_{\mathrm{lp}}  \right\rbrace :=$ least preferred transmitter with power level $l_{\mathrm{lp}}$ assigned to $n_{\mathrm{mp}}$.
\STATE Set $x_{k_{\mathrm{lp}}}^{\left( n_{\mathrm{mp}}, l_{\mathrm{lp}}\right)} := 0$. ~\COMMENT{\footnotesize Revoke assignment due to  interference threshold violation}

\STATE $\mathfrak{I}^{(n_\mathrm{mp})} := \sum\limits_{\substack{ k^\prime =1, }}^{K}\sum\limits_{l^\prime = 1}^{L}x_{k^\prime}^{(n_\mathrm{mp}, l^\prime)} g_{k^\prime,m_{k^\prime}^*}^{(n_\mathrm{mp})} p_{k^\prime}^{(n_\mathrm{mp})}$. ~\COMMENT{\footnotesize Update interference level} 

 { \footnotesize \textit{/* Update preference profiles */} }
\FORALL {successor $\lbrace \hat{k}_{\mathrm{lp}}, \hat{l}_{\mathrm{lp}}  \rbrace$ of $ \left\lbrace k_{\mathrm{lp}}, l_{\mathrm{lp}}  \right\rbrace$ on profile $\boldsymbol{\mathscr{P}}_{n_{\mathrm{mp}}}(\mathcal{K}^{\mathrm T}, \mathcal{L})$} 

\STATE remove $\lbrace \hat{k}_{\mathrm{lp}}, \hat{l}_{\mathrm{lp}}\rbrace$ from $\boldsymbol{\mathscr{P}}_{n_{\mathrm{mp}}}(\mathcal{K}^{\mathrm T}, \mathcal{L})$.
\STATE remove $\left\lbrace n_{\mathrm{mp}}, l_{\mathrm{mp}} \right\rbrace$ from $\boldsymbol{\mathscr{P}}_{\hat{k}_{\mathrm{lp}}}(\mathcal{N}, \mathcal{L})$. 
\ENDFOR

\UNTIL{$\mathfrak{I}^{(n_\mathrm{mp})} < I_{\mathrm{max}}^{(n_\mathrm{mp})}$ }

\ENDIF 

\ENDWHILE

\end{algorithmic}
\end{algorithm}

 The overall stable matching based resource allocation approach is summarized in \textbf{Algorithm \ref{alg:ra_sm}}. Note that \textbf{Algorithm \ref{alg:ta_sm}} is executed repeatedly. The convergence of \textbf{Algorithm \ref{alg:ra_sm}} occurs when the outcome of two consecutive local matching is similar, e.g., $\mathbf{X}(t) = \mathbf{X}(t-1)$ and as a consequence $R(t) = R(t-1)$, where $R(t) = \sum\limits_{k=1}^{K} R_{u_k}(t)$ denotes the achievable sum rate of the underlay tier at iteration $t$.


\begin{algorithm} [!t]
\AtBeginEnvironment{algorithmic}{\small} 
\caption{Stable matching-based resource allocation}
\label{alg:ra_sm}
\begin{algorithmic}[1]   
\renewcommand{\algorithmicrequire}{\textbf{Input:}}
\renewcommand{\algorithmicensure}{\textbf{Output:}}
\renewcommand{\algorithmicforall}{\textbf{for each}}
\renewcommand{\algorithmiccomment}[1]{\textit{/* #1 */}}

\renewcommand{\algorithmicensure}{\textbf{Initialization:}}
\ENSURE

\STATE Estimate the CSI parameters from previous time slot.

\STATE Each underlay transmitter $k \in \mathcal{K}^{\mathrm T}$ randomly selects a transmission alignment and the MBS broadcasts the aggregated interference of each RB using pilot signals.

\STATE Each underlay transmitter $k \in \mathcal{K}^{\mathrm T}$ builds the preference profile $\boldsymbol{\mathscr{P}}_{k}(\mathcal{N}, \mathcal{L})$ from the CSI estimations and the utility function given by Equation  (\ref{eq:sm_utility}).

\STATE For each $n \in \mathcal{N}$,  the MBS builds the preference profiles $\boldsymbol{\mathscr{P}}_{n}(\mathcal{K}^{\mathrm T}, \mathcal{L})$.

\STATE Initialize number of iterations $t := 1$.

\renewcommand{\algorithmicensure}{\textbf{Update:}}
\vspace*{0.5em}
\ENSURE



\WHILE{$\mathbf{X}(t) \neq \mathbf{X}(t-1)$ \AND $t$ is less than some predefined threshold $T_{\mathrm{max}}$}

\STATE MBS obtains a local stable matching $\mathbf{X}(t)$ using \textbf{Algorithm \ref{alg:ta_sm}}, calculates the aggregated interference $I^{(n)}(t)$ for $\forall n$ and informs the transmitters.

\STATE  Each underlay transmitter $k \in \mathcal{K}^{\mathrm T}$ updates the preference profile $\boldsymbol{\mathscr{P}}_{k}(\mathcal{N}, \mathcal{L})$ based on current allocation vector $\mathbf{X}(t)$ and interference level $I^{(n)}(t)$.

\STATE MBS updates the preference profile $\boldsymbol{\mathscr{P}}_{n}(\mathcal{K}^{\mathrm T}, \mathcal{L})$ for $\forall n \in \mathcal{N}$ using  $\mathbf{X}(t)$ and $I^{(n)}(t)$.

\STATE Update $t := t+1$.
\ENDWHILE

\renewcommand{\algorithmicensure}{\textbf{Allocation:}}
\vspace*{0.5em}
\ENSURE

\STATE Allocate the RB and power levels to the SBSs and D2D UEs based on the matching obtained from the update phase. 

\end{algorithmic}
\end{algorithm}

\subsection{Stability, Optimality, and Complexity of the Solution}

In this section, we  analyze the solution obtained by stable matching approach. The stability, optimality, and the complexity of the algorithm are discussed in the following.

\subsubsection{Stability}

The notion of stability in the matching $\mu$ means that none of the agents (e.g., either underlay transmitters or the resources) prefers to change the allocation obtained by $\mu$. Hence, the matching $\mu$ is stable if no transmitter and no resource who are not allocated to each other, 
as given in $\mu$, prefer each other to their allocation in $\mu$. The transmitters and resources are said to be \textit{acceptable} if the agents (e.g., transmitters and resources) prefer each other to remain unallocated. In addition, a matching $\mu$ is called \textit{individually rational} if no agent  $\tilde{\jmath}$ prefers unallocation to the matching in $\mu(\tilde{\jmath})$. Before formally defining the stability of matching, we introduce the term \textit{blocking pair} which is defined as

\begin{definition} \label{def:blocking}
A matching $\mu$ is \textbf{blocked} by a pair of agent $(i,j)$ if they prefer each other to the matching obtain by $\mu$, i.e.,  $i \succeq_{j} \mu(j)$ and $j \succeq_{i} \mu(i)$.
\end{definition}

Using the above definition, the stability of the matching can be defined as follows \cite[Chapter 5]{matchingbook_two}:

\begin{definition}  \label{def:stability}
A matching $\mu$ is \textbf{stable} if it is individually rational and there is no tuple $(k, n, l)$ within the set of acceptable agents such that $k$ prefers $\lbrace n,l \rbrace$ to  $\mu(k)$ and $n$ prefers $\lbrace k,l \rbrace$ to  $\mu(n)$, i.e., not blocked by any pair of agents. 
\end{definition}

The following theorem shows that the solution obtained by the matching algorithm is stable.

\begin{theorem} \label{thm:stability}
The assignment performed in \textbf{Algorithm \ref{alg:ta_sm}} leads to a stable allocation.
\end{theorem} 
\begin{proof}
We proof the theorem by contradiction. Let $\mu$ be a matching obtained by \textbf{Algorithm \ref{alg:ta_sm}}. Let us assume that the resource $\lbrace n, l \rbrace$ is not allocated to the transmitter $k$, but it belongs to a higher order in the preference list. According to this assumption,
the tuple $(k, n, l)$ will block $\mu$. Since the position of the resource $\lbrace n, l \rbrace$ in the preference profile of $k$ is higher compared to any resource $\lbrace \hat{n}, \hat{l} \rbrace$ that is matched by $\mu$, i.e., $\lbrace n, l \rbrace \succeq_{k} \mu(k) $, transmitter $k$ must select $\lbrace n, l \rbrace$ before the algorithm terminates. Note that, the resource $\lbrace n,l \rbrace$ is not assigned to transmitter $k$ in the matching outcome $\mu$. This implies that 
$k$ is unassigned with the resource $\lbrace n , l \rbrace$ (e.g., line 9 in \textbf{Algorithm \ref{alg:ta_sm}}) and $(k, \hat{n}, \hat{l})$ is a better assignment. As a result, the tuple $(k, n, l)$ will not block $\mu$, which contradicts our assumption. The proof concludes since no blocking pair exists, and therefore, the matching outcome $\mu$ leads to a stable matching.
\end{proof}

It is worth mentioning that the assignment is stable at each iteration of \textbf{Algorithm \ref{alg:ta_sm}}. Since after  evaluation of the utility, the preference profiles are updated and the matching subroutine is repeated, a stable allocation is obtained at each iteration.

\subsubsection{Optimality}

The optimality property of the stable matching approach can be observed using the definition of weak Pareto optimality. Let $\mathcal{R}_{\mu}$ denote the sum-rate obtained by matching $\mu$. A matching $\mu$ is weak Pareto optimal if there is no other matching $ \widehat{\mu}$ that can achieve a better sum-rate, i.e., $\mathcal{R}_{\widehat{\mu}} \geq \mathcal{R}_{\mu} $ \cite{sm_def}.

\begin{theorem}
The stable matching-based resource allocation algorithm is weak Pareto optimal.
\end{theorem}
\begin{proof}
Let us consider $\mu$ to be the stable allocation obtained by \textbf{Algorithm \ref{alg:ta_sm}}. For instance, let $\widehat{\mu}$ be an arbitrary stable outcome better that $\mu$, i.e., $\widehat{\mu}$ can achieve a better sum-rate. Since the allocation $\widehat{\mu}$ is better than $\mu$, there exists atleast one resource $\lbrace \hat{n}, \hat{l} \rbrace$ allocated to transmitter $k$ in $\widehat{\mu}$, and $k$ is allocated to the resource $\lbrace n, l \rbrace$ in $\mu$. According to our assumption, $k$ prefers $\lbrace \hat{n}, \hat{l} \rbrace$ to $\lbrace n,l \rbrace$, and let $\lbrace \hat{n}, \hat{l} \rbrace$ be allocated to transmitter $\hat{k}$ in $\mu$. It is obvious that resource  $\lbrace \hat{n}, \hat{l} \rbrace$ is better than $\lbrace n,l \rbrace$ to transmitter $k$ and $\lbrace k, l \rbrace$ is better than $\lbrace \hat{k}, \hat{l} \rbrace$ to resource $\hat{n}$, i.e., $\lbrace \hat{n}, \hat{l} \rbrace \succeq_k \lbrace n,l \rbrace$ 
and $\lbrace k, l \rbrace \succeq_{\hat{n}} \lbrace  \hat{k}, \hat{l} \rbrace$. 
By the definition of blocking pair, $\mu$ is blocked by $(k, \hat{n}, \hat{l})$ and hence $\mu$ is unstable. This contradicts our assumption that $\mu$ is a stable allocation. Since there is no stable outcome $\widehat{\mu}$ which is better that $\mu$, by definition  $\mu$ is an optimal allocation.
\end{proof}

\subsubsection{Complexity}

It is possible to show that the stable matching algorithm will iterate for finite number of times. 

\begin{theorem}
\label{thm:sm_time}
The RB allocation subroutine terminates after some finite step $T^\prime$.
\end{theorem} 
\begin{proof}
Let the finite set $\tilde{\mathcal{X}}$ represent the all possible combinations of transmitter-resource matching where each element  $\tilde{x}_{k}^{(n,l)} \in \tilde{\mathcal{X}}$ denotes the resource $\lbrace n,l \rbrace$ is allocated to the transmitter $k$. Since no transmitter is rejected by the same resource more than once (i.e., line 9 in \textbf{Algorithm \ref{alg:ta_sm}}), the finiteness of the set $\tilde{\mathcal{X}}$ ensures the termination of the matching subroutine in finite number of steps.
\end{proof}

For each underlay transmitter, the complexity to build the preference profile  using any standard sorting algorithm is $\mathcal{O}\left( NL \log( NL) \right)$ (line 8, \textbf{Algorithm \ref{alg:ra_sm}}). Similarly, in line 9, the complexity to output the ordered set of preference profile for the RBs is of  $\mathcal{O}\left( N KL \log (KL)  \right)$. Let $\xi = \displaystyle \sum_{k = 1 }^{K} |\boldsymbol{\mathscr{P}}_{k}(\mathcal{N}, \mathcal{L}) | + \sum_{n = 1}^{N} |\boldsymbol{\mathscr{P}}_{n}(\mathcal{K}^{\mathrm T}, \mathcal{L}) |  = 2 K N L$ be the total length of input preferences in \textbf{Algorithm \ref{alg:ta_sm}}, where $|\boldsymbol{\mathscr{P}}_j(\cdot)|$ denotes the length of profile vector $\boldsymbol{\mathscr{P}}_j(\cdot)$. From \textbf{Theorem \ref{thm:sm_time}} and \cite[Chapter 1]{matching_thesis} it can be shown that, if implemented with suitable data structures, the time complexity of the RB allocation subroutine is linear in the size of input preference profiles, i.e., $\mathcal{O}(\xi) = \mathcal{O}\left( K N L \right)$. Since the update phase of \textbf{Algorithm \ref{alg:ra_sm}} runs at most fixed $T < T_{\mathrm{max}}$ iterations, the complexity of the stable matching-based solution is linear in $K, N, L$.

\section{Message Passing Approach for Resource Allocation} \label{sec:mp_ra}

In the following, we reformulate the resource allocation problem $\mathbf{P \ref{opt:combopt}}$ in such a way that can be solved with a message passing (MP) technique. The MP approach involves computation of the \textit{marginals}, e.g., the messages exchanged between the nodes of a specific graphical model. Among different representations of graphical model, we consider \textit{factor graph} based MP scheme. A factor graph is made up of two different types of nodes, i.e., \textit{function} and \textit{variable} nodes, and an edge connects a function (e.g., factor) node to a variable node if and only if the variable appears in the function. Mathematically, this can be expressed as follows \cite{factorgraph_thory}: 

\begin{definition}
A factor graph can be represented by a $\mathcal{V}$-$\mathcal{F}$ bipartite graph where $\mathcal{V} = \left\lbrace v_1, \cdots v_a \right\rbrace$ is the set of variable nodes and $\mathcal{F} = \left\lbrace f_1(\cdot), \cdots f_b(\cdot) \right\rbrace$ is the set of function (e.g., factor) nodes. The connectivity (e.g., edges) of the factor graph can be represented by an $a \times b$ binary matrix $\mathbf{E} = [E_{i,j}]$ where $E_{i,j} = 1$ if the variable node $i$ is connected with the factor node $j$ and $E_{i,j} = 0$, otherwise.
\end{definition}

\subsection{Overview of the MP Scheme} \label{subsec:mp_overview}

Before presenting the details resource allocation approach for a heterogeneous scenario, we briefly introduce the generic MP scheme (for the details of factor graph based MP scheme refer to \cite{factorgraph_thory}). Let us consider the maximization of an arbitrary function $f(v_1, \cdots, v_J)$ over all possible values of the argument, i.e., $Z = \underset{\mathbf{v}}{\operatorname{max}} ~ f(\mathbf{v})$ where $\mathbf{v} = \left[ v_1, \cdots, v_J \right]^{\mathsf{T}}$. We denote by $\underset{\mathbf{v}}{\operatorname{max}}$ that the maximization is computed over all possible combinations of the elements of the the vector $\mathbf{v}$. The \textit{marginal} of $Z$ with respect to variable $v_j$ is given by $\phi_j(v_j) =  \underset{\sim (v_j)}{\operatorname{max}} ~f(\mathbf{v})$ where $\underset{\sim (\cdot)}{\operatorname{max}}$ denote the maximization over all variables except $(\cdot)$. Let us now decompose $f(\mathbf{v})$ into summation of $I$ functions, i.e., $\sum\limits_{i=1}^{I} f_{i}(\hat{v}_i)$ where $\hat{v}_i$ is a subset of the elements of the vector $\mathbf{v}$ and let $\mathbf{f}= \left[ f_{1}(\cdot), \cdots, f_{I}(\cdot)  \right]^{\mathsf{T}}$ is the vector of $I$ functions. In addition, let  $\mathfrak{f}_j$ represents subset of functions in $\mathbf{f}$ where the variable $v_j$ appears. Hence the marginal can be rewritten as $\phi_j(v_j) = \underset{\sim (v_j)}{\operatorname{max}} ~\sum\limits_{i=1}^{I} f_{i}(\hat{v}_i)$. According to the \textit{max-sum} MP strategy the message passed by any variable node $v_j$ to any generic function node $f_i(\cdot)$ is given by $\delta_{v_j \rightarrow f_i(\cdot)}(v_j) = \sum\limits_{i^\prime \in \mathfrak{f}_j, i^\prime \neq i} \delta_{f_{i^\prime}(\cdot)  \rightarrow v_j }(v_j)$. Similarly, the message from function node $f_i{(\cdot)}$ to variable node $v_j$ is given as $\delta_{f_i(\cdot) \rightarrow v_j}(v_j) = \underset{\sim (v_j)}{\operatorname{max}} \left(f_i(v_1, \cdots, v_J) + \sum\limits_{j^\prime \in \hat{v}_i, j^\prime \neq j} \delta_{v_{j^\prime}   \rightarrow f_{i}(\cdot)}(v_{j^\prime}) \right) $. When the factor graph is cycle free (e.g., there is a unique path connecting any two nodes), all the variables nodes $j = \lbrace 1, \cdots, J \rbrace$ can compute the marginals as $\phi_j(v_j) = \sum\limits_{i=1}^{I} \delta_{f_{i}(\cdot)  \rightarrow v_j }(v_j)$. Utilizing the general distributive law (e.g., $\operatorname{\max} \sum = \sum  \operatorname{\max}$) \cite{mp_distributive} the maximization therefore can be computed as $Z = \sum\limits_{j=1}^{J} \underset{v_j}{\operatorname{max}} ~\phi_j(v_j)$.

\subsection{Reformulation of the Resource Allocation Problem Utilizing MP Approach}

In order to solve the resource allocation problem $\mathbf{P \ref{opt:combopt}}$ presented in Section \ref{subsec:rap_org} using MP, we reformulate it as a utility maximization problem. Let us define the reward functions $\mathfrak{W}_n(\mathbf{X})$ and   $\mathfrak{R}_k(\mathbf{X})$  where the transmission alignment vector $\mathbf{X}$ is given by Equation (\ref{eq:rap_X}).  With the constraint in Equation (\ref{eq:opt_intf}), we can define $\mathfrak{W}_n(\mathbf{X})$ as follows:
\begin{equation} \label{eq:util_mp1}
\mathfrak{W}_n(\mathbf{X}) = 
\begin{cases}  0, & \text{if~} \sum\limits_{k =1}^{K}\sum\limits_{l = 1}^{L}x_{k}^{(n, l)} g_{k,m_k^*}^{(n)} p_k^{(n)} < I_{\mathrm{max}}^{(n)} \\ 
- \infty,  & \text{otherwise.}\end{cases}
\end{equation}
Similarly to deal with the constraint in Equation (\ref{eq:opt_rbpw})  we define $\mathfrak{R}_k(\mathbf{X})$ as
\begin{equation} \label{eq:util_mp2}
\mathfrak{R}_k(\mathbf{X}) = 
\begin{cases}  \sum\limits_{n = 1}^{N} \sum\limits_{l = 1}^{L}  ~x_{k}^{(n,l)}    B_{\mathrm {RB}}  \log_2\left( 1 + \gamma_{u_k}^{(n)} \right) & \text{if~} \sum\limits_{n = 1}^{N} \sum\limits_{l = 1}^{L} x_{k}^{(n,l)} \leq 1 \\ 
- \infty  & \text{otherwise.}\end{cases}
\end{equation}
The interpretations of the reward functions in Equations (\ref{eq:util_mp1}) and (\ref{eq:util_mp2}) are straightforward. Satisfying the interference constraint in Equation (\ref{eq:opt_intf}) does not cost any penalty (e.g., zero reward) in the function $\mathfrak{W}_n(\mathbf{X})$, and in the function $\mathfrak{R}_k(\mathbf{X})$ fulfillment of the RB requirement constraint in Equation (\ref{eq:opt_rbpw}) gives the desired data rate. However, both in the functions $\mathfrak{W}_n(\mathbf{X})$ and $\mathfrak{R}_k(\mathbf{X})$, the unfulfilled constraints, respectively, given by in Equations (\ref{eq:opt_intf}) and (\ref{eq:opt_rbpw}), result in infinite cost.  

From the Equations (\ref{eq:util_mp1}) and (\ref{eq:util_mp2}), the resource allocation problem $\mathbf{P \ref{opt:combopt}}$ can be rewritten as
$$\underset{\mathbf{X}}{\operatorname{max}} \left( \sum\limits_{n=1}^{N} \mathfrak{W}_n(\mathbf{X}) + \sum\limits_{k=1}^{K} \mathfrak{R}_k(\mathbf{X})  \right)$$
and the optimal transmission allocation vector is therefore given by 
\begin{equation} \label{eq:mp_Xall}
\mathbf{X}^* = \underset{\mathbf{X}}{\operatorname{argmax}} \left( \sum\limits_{n=1}^{N} \mathfrak{W}_n(\mathbf{X}) + \sum\limits_{k=1}^{K} \mathfrak{R}_k(\mathbf{X})  \right).
\end{equation}
Since our goal is to obtain a distributed solution for the above resource allocation problem, we focus on a single transmission alignment allocation variable, e.g., $x_{k}^{(n,l)}$. From Equation (\ref{eq:mp_Xall}) we obtain ${x_{k}^{(n,l)}}^* = \underset{x_{k}^{(n,l)}} {\operatorname{argmax}}~ \phi_{k}^{(n,l)}\big(x_{k}^{(n,l)}\big)$ where the marginal $\phi_{k}^{(n,l)} \big(x_{k}^{(n,l)} \big)$ is given by
\begin{equation} \label{eq:marginal_mp}
\phi_{k}^{(n,l)}\big(x_{k}^{(n,l)}\big) = \underset{\sim \bigl(x_{k}^{(n,l)} \bigl)}{\operatorname{max}}  \left( \sum\limits_{n=1}^{N} \mathfrak{W}_n(\mathbf{X}) + \sum\limits_{k=1}^{K} \mathfrak{R}_k(\mathbf{X})  \right).
\end{equation} 
As mentioned in the previous section, $\underset{\sim \bigl(x_{k}^{(n,l)} \bigl)}{\operatorname{max}}$ denote the maximization over all variables in $\mathbf{X}$ except $x_{k}^{(n,l)}$. The marginalization in Equation (\ref{eq:marginal_mp}) can be computed in a distributed way where each node conveys the solution of a local problem to one another by passing information messages according to the max-sum MP strategy. Note that according to our system model the underlay transmitters and the resources (e.g., transmission alignments) can form a bipartite graph, e.g., each transmission alignment $\lbrace n,l \rbrace$ can be assigned to any of the $K$ transmitters as long as interference to the MUEs on RB $n$ is below threshold. Without loss of generality, let us consider a generic transmission alignment, e.g., RB-power level pair $\lbrace n,l \rbrace \in \mathcal{N} \times \mathcal{L}$ and an underlay transmitter $k \in \mathcal{K}^{\mathrm T}$. Using the function in Equation (\ref{eq:util_mp1})  and utilizing the max-sum MP strategy presented in Section \ref{subsec:mp_overview}, it is possible to show that the message delivered by the resource $\lbrace n,l \rbrace$ to the transmitter $k$ can be expressed as \cite{min-sum-mp}
\begin{equation} \label{eq:mp_msg1}
\begin{aligned}
\delta_{\lbrace n,l \rbrace \rightarrow k} \big( x_{k}^{(n,l)} \big) = \operatorname{max} \sum\limits_{k^\prime \in \mathcal{K}^{\mathrm T},~ k^\prime \neq k}  \delta_{k^\prime \rightarrow \lbrace n,l \rbrace} \big( x_{k^\prime}^{(n,l)} \big) \\
\text{subject to:~~ } \sum\limits_{k =1}^{K}\sum\limits_{l = 1}^{L}x_{k}^{(n, l)} g_{k,m_k^*}^{(n)} p_k^{(n)} < I_{\mathrm{max}}^{(n)}.
\end{aligned}
\end{equation}
Note that the term $\delta_{k \rightarrow \lbrace n,l \rbrace} \big( x_{k}^{(n,l)} \big)$ in the above equation denotes the message from transmitter $k$ to the resource $\lbrace n,l \rbrace$ which can be written as \cite{min-sum-mp}
\begin{equation} \label{eq:mp_msg2}
\begin{aligned}
\delta_{k \rightarrow \lbrace n,l \rbrace} \big( x_{k}^{(n,l)} \big) = x_{k}^{(n,l)}    R_{u_k}^{(n,l)} + \operatorname{max}  \sum\limits_{\substack{\lbrace n^\prime, l^\prime \rbrace \in \mathcal{N} \times \mathcal{L} \\ n^\prime \neq n,~ l^\prime \neq l }} x_{k}^{(n^\prime,l^\prime)} &   R_{u_k}^{(n^\prime,l^\prime)} + \delta_{\lbrace n^\prime,l^\prime \rbrace \rightarrow k} \big( x_{k}^{(n^\prime,l^\prime)} \big) \\
\text{subject to:~~ } \sum\limits_{n = 1}^{N} \sum\limits_{l = 1}^{L} x_{k}^{(n,l)} &\leq 1
\end{aligned}
\end{equation}
where $R_{u_k}^{(n,l)} =  B_{\mathrm {RB}}  \log_2\left( 1 + \Gamma_{k}^{(n,l)} \right)$ and $\Gamma_k^{(n, l)} \triangleq {\gamma_{u_k}^{(n)}}_{ \!\! \vert p_k^{(n)} = l }$.

The interpretation of the Equations (\ref{eq:mp_msg1}) and  (\ref{eq:mp_msg2}) are as follows: the messages $\delta_{\lbrace n,l \rbrace \rightarrow k} ( 1 )$ and $\delta_{k \rightarrow \lbrace n,l \rbrace} (1)$ carry the information relative to the use of the resource $\lbrace n, l\rbrace$ by the transmitter $k$; while the messages $\delta_{\lbrace n,l \rbrace \rightarrow k} (0)$ and $\delta_{k \rightarrow \lbrace n,l \rbrace} (0)$ carry the information relative to the lack of transmission over the resource $\lbrace n, l\rbrace$ by the transmitter $k$. In order to obtain both the messages $\delta_{\lbrace n,l \rbrace \rightarrow k} \big( x_{k}^{(n,l)} \big) $ and $\delta_{k \rightarrow \lbrace n,l \rbrace} \big( x_{k}^{(n,l)} \big) $, it is required to solve the local optimization problem relative to the allocation variable $ x_{k}^{(n,l)} $.

Based on the discussions of Section \ref{subsec:mp_overview}, the link-wise marginal in Equation (\ref{eq:marginal_mp}) can be written as \cite{min-sum-mp}
\begin{equation} \label{eq:mp_marginal_modified}
\phi_{k}^{(n,l)}\big(x_{k}^{(n,l)}\big) = \delta_{\lbrace n,l \rbrace \rightarrow k} \big( x_{k}^{(n,l)} \big) +  \delta_{k \rightarrow \lbrace n,l \rbrace} \big( x_{k}^{(n,l)} \big) 
\end{equation}
and hence the transmission allocation variable is given by 
\begin{equation} \label{eq:mp_allocation_upd}
{x_{k}^{(n,l)}}^* = \underset{x_{k}^{(n,l)}} {\operatorname{argmax}}~ \phi_{k}^{(n,l)}\big(x_{k}^{(n,l)}\big).
\end{equation}
At each iteration of the MP-based resource allocation algorithm, at most one message passes through the edge of any given direction (e.g., from transmitters to resources or from resources to transmitters); and each iteration the messages are updated by replacing the previous message sent on the same edge in the same direction \cite{min-sum-mp}. When both the messages given by Equations (\ref{eq:mp_msg1}) and (\ref{eq:mp_msg2}) are available, the marginal can be computed using Equation (\ref{eq:mp_marginal_modified}) and the transmission allocation variable is obtained by Equation (\ref{eq:mp_allocation_upd}).

\subsection{Effective Implementation of MP Scheme in a Practical Heterogeneous Network} \label{subsec:mp_effective}

It is worth noting that, sending messages from resources to transmitters (and vice versa) requires actual transmission on the radio channel. In a practical LTE-A-based 5G system, since the exchange of messages actually involves effective transmissions over the channel, the MP scheme described in the preceding section might be limited by the signaling overhead  due to transfer of messages between the transmitters and resources. In the following, we observe that the amount of message signaling can be significantly reduced by some algebraic manipulations. Since the messages carry the information regarding whether any resource is used by any underlay transmitter, each transmitter $k$ actually delivers a real valued vector with two element, i.e., $\boldsymbol{\delta}_{k \rightarrow \lbrace n,l \rbrace} = \left[ \delta_{k \rightarrow \lbrace n,l \rbrace} (1),~ \delta_{k \rightarrow \lbrace n,l \rbrace} (0) \right]^{\mathsf{T}}$ and each resource $\lbrace n, l\rbrace$ delivers the vector $\boldsymbol{\delta}_{\lbrace n,l \rbrace \rightarrow k} = \left[ \delta_{\lbrace n,l \rbrace \rightarrow k} (1),~ \delta_{\lbrace n,l \rbrace \rightarrow k} (0) \right]^{\mathsf{T}}$. Let us now rewrite the message $\delta_{k \rightarrow \lbrace n,l \rbrace} \big( x_{k}^{(n,l)} \big)$ using the utility function introduced in Equation (\ref{eq:sm_utility}) as follows:

\begin{equation} \label{eq:mp_msg2_wUtil}
\delta_{k \rightarrow \lbrace n,l \rbrace} \big( x_{k}^{(n,l)} \big) = x_{k}^{(n,l)}    \mathfrak{U}_{k}^{(n,l)} + \operatorname{max}  \sum\limits_{\substack{\lbrace n^\prime, l^\prime \rbrace \in \mathcal{N} \times \mathcal{L} \\ n^\prime \neq n,~ l^\prime \neq l }} x_{k}^{(n^\prime,l^\prime)}  \mathfrak{U}_{k}^{(n^\prime,l^\prime)} + \delta_{\lbrace n^\prime,l^\prime \rbrace \rightarrow k} \big( x_{k}^{(n^\prime,l^\prime)} \big). 
\end{equation}
By subtracting the constant term $\sum\limits_{\substack{\lbrace n^\prime, l^\prime \rbrace \in \mathcal{N} \times \mathcal{L} \\ n^\prime \neq n,~ l^\prime \neq l }} \delta_{\lbrace n^\prime,l^\prime \rbrace \rightarrow k} (0) $ from the both sides of Equation (\ref{eq:mp_msg2_wUtil}) we can obtain the following:
\begin{equation} \label{eq:mp_msg2_wUtil_sub}
\begin{aligned}
& \delta_{k \rightarrow \lbrace n,l \rbrace} \big( x_{k}^{(n,l)} \big) - \sum\limits_{\substack{\lbrace n^\prime, l^\prime \rbrace \in \mathcal{N} \times \mathcal{L} \\ n^\prime \neq n,~ l^\prime \neq l }} \delta_{\lbrace n^\prime,l^\prime \rbrace \rightarrow k} (0) = x_{k}^{(n,l)}    \mathfrak{U}_{k}^{(n,l)} ~~+\\
&   \operatorname{max}  \sum\limits_{\substack{\lbrace n^\prime, l^\prime \rbrace \in \mathcal{N} \times \mathcal{L} \\ n^\prime \neq n,~ l^\prime \neq l }} x_{k}^{(n^\prime,l^\prime)}  \mathfrak{U}_{k}^{(n^\prime,l^\prime)} + \delta_{\lbrace n^\prime,l^\prime \rbrace \rightarrow k} \big( x_{k}^{(n^\prime,l^\prime)} \big) - \delta_{\lbrace n^\prime,l^\prime \rbrace \rightarrow k} (0).
\end{aligned}
\end{equation}

Let us now introduce the parameter $\psi_{\lbrace n,l \rbrace \rightarrow k} = \delta_{\lbrace n,l \rbrace \rightarrow k} (1) - \delta_{\lbrace n,l \rbrace \rightarrow k} (0)$ defined as the normalized message.  For instance, consider the vector $$\Psi_{k} = \left[ \mathfrak{U}_{k}^{(1,1)} + \psi_{\lbrace 1,1 \rbrace \rightarrow k}, \cdots, \mathfrak{U}_{k}^{(1,L)} + \psi_{\lbrace 1,L \rbrace \rightarrow k}, \cdots, \mathfrak{U}_{k}^{(N,L)} + \psi_{\lbrace N,L \rbrace \rightarrow k} \right]^{\mathsf{T}}$$ and let us denote by $\left\langle \upsilon_{\lbrace n^\prime,l^\prime \rbrace \rightarrow k} \right\rangle_{\sim \lbrace n,l \rbrace}$  the maximal entry of the vector $\Psi_{k}$ without considering the term $\mathfrak{U}_{k}^{(n,l)} + \psi_{\lbrace n,l \rbrace \rightarrow k}$. It can be noted that the terms within the summation in Equation (\ref{eq:mp_msg2_wUtil_sub}) are either $0$ (e.g., when $x_{k}^{(n,l)} = 0$) or $\mathfrak{U}_{k}^{(n^\prime,l^\prime)} + \psi_{\lbrace n^\prime,l^\prime \rbrace \rightarrow k}$ (e.g., when $x_{k}^{(n,l)} = 1$).  Since each transmitter requires only a single transmission alignment, when the variable  $x_{k}^{(n,l)} = 0$, only one term in the summation of Equation (\ref{eq:mp_msg2_wUtil_sub}) is non-zero. For the case $x_{k}^{(n,l)} = 1$,  no term within the summation of Equation (\ref{eq:mp_msg2_wUtil_sub}) is non-zero. Consequently, for $x_{k}^{(n,l)} = 0$, the maximum rate will be achieved if 
\begin{equation} \label{eq:mp_msg_wSort1}
\delta_{k \rightarrow \lbrace n,l \rbrace} (0) - \sum\limits_{\substack{\lbrace n^\prime, l^\prime \rbrace \in \mathcal{N} \times \mathcal{L} \\ n^\prime \neq n,~ l^\prime \neq l }} \delta_{\lbrace n^\prime,l^\prime \rbrace \rightarrow k} (0) =  \left\langle \upsilon_{\lbrace n^\prime,l^\prime \rbrace \rightarrow k} \right\rangle_{\sim \lbrace n,l \rbrace}.
\end{equation}
Similarly, when $x_{k}^{(n,l)} = 1$, the maximum is given by
\begin{equation} \label{eq:mp_msg_wSort2}
\delta_{k \rightarrow \lbrace n,l \rbrace} (1) - \sum\limits_{\substack{\lbrace n^\prime, l^\prime \rbrace \in \mathcal{N} \times \mathcal{L} \\ n^\prime \neq n,~ l^\prime \neq l }} \delta_{\lbrace n^\prime,l^\prime \rbrace \rightarrow k} (0) = \mathfrak{U}_{k}^{(n,l)}.
\end{equation}
Since by definition $\psi_{ k \rightarrow \lbrace n,l \rbrace } = \delta_{k \rightarrow \lbrace n,l \rbrace} (1) - \delta_{k \rightarrow \lbrace n,l \rbrace} (0)$, from the Equations (\ref{eq:mp_msg_wSort1}) and  (\ref{eq:mp_msg_wSort2}), the normalized messages from the transmitter $k$ to the resource $\lbrace n, l\rbrace$ can be derived as
\begin{align} \label{eq:mp_msg_norm1}
\psi_{ k \rightarrow \lbrace n,l \rbrace } &= \mathfrak{U}_{k}^{(n,l)} - \left\langle \upsilon_{\lbrace n^\prime,l^\prime \rbrace \rightarrow k} \right\rangle_{\sim \lbrace n,l \rbrace} \nonumber \\
&= \mathfrak{U}_{k}^{(n,l)} - \left\langle \mathfrak{U}_{k}^{(n^\prime,l^\prime)} + \psi_{\lbrace n^\prime,l^\prime \rbrace \rightarrow k} \right\rangle_{\sim \lbrace n,l \rbrace}.
\end{align}
Likewise, from \cite{min-sum-mp}, it can be shown that the normalized message sent from the resource $\lbrace n, l\rbrace$ to the transmitter $k$ becomes 
\begin{equation} \label{eq:mp_msg_norm2}
\psi_{\lbrace n,l \rbrace \rightarrow k} = \delta_{\lbrace n,l \rbrace \rightarrow k} (1) - \delta_{\lbrace n,l \rbrace \rightarrow k} (0) = - \underset{k^\prime \in \mathcal{K}^{\mathrm T}, k^\prime \neq k}{\operatorname{max}} \psi_{ k^\prime \rightarrow \lbrace n,l \rbrace }. 
\end{equation}
For any arbitrary graph, the allocation variables may keep oscillating and might not converge to any fixed point, and the MP scheme may require some heuristic approach to terminate. However, in the context of loopy graphical models, by introducing a suitable weight, the messages given by Equations (\ref{eq:mp_msg_norm1}) and (\ref{eq:mp_msg_norm2}) perturb to a fixed point \cite{min-sum-mp, remp_proof}. Accordingly, Equations (\ref{eq:mp_msg_norm1}) and (\ref{eq:mp_msg_norm2}) can be rewritten as \cite{min-sum-mp}
\begin{align}
\psi_{ k \rightarrow \lbrace n,l \rbrace } &= \mathfrak{U}_{k}^{(n,l)} - \omega \left\langle \mathfrak{U}_{k}^{(n^\prime,l^\prime)} + \psi_{\lbrace n^\prime,l^\prime \rbrace \rightarrow k} \right\rangle_{\sim \lbrace n,l \rbrace} - (1-\omega) \left( \mathfrak{U}_{k}^{(n,l)} + \psi_{\lbrace n,l \rbrace \rightarrow k} \right) \label{eq:mp_norm_msg_w1} \\
\psi_{\lbrace n,l \rbrace \rightarrow k} &= - \omega \underset{k^\prime \in \mathcal{K}^{\mathrm T}, k^\prime \neq k}{\operatorname{max}} \psi_{ k^\prime \rightarrow \lbrace n,l \rbrace } - (1-\omega)~ \psi_{ k \rightarrow \lbrace n,l \rbrace } \label{eq:mp_norm_msg_w2}
\end{align}
where $\omega \in (0, 1]$ denotes the weighting factor for each edge. Notice that when $\omega = 1$, the messages given by Equations (\ref{eq:mp_norm_msg_w1}) and (\ref{eq:mp_norm_msg_w2}) reduce to the original formulation, e.g., Equations (\ref{eq:mp_msg_norm1}) and (\ref{eq:mp_msg_norm2}), respectively. Given the normalized messages $\psi_{ k \rightarrow \lbrace n,l \rbrace }$ and $\psi_{\lbrace n,l \rbrace \rightarrow k} $ for $\forall k, n, l$, the node marginals for the normalized messages can be calculated as $\tau_{k}^{(n,l)} = \psi_{ k \rightarrow \lbrace n,l \rbrace } + \psi_{\lbrace n,l \rbrace \rightarrow k} $ and hence from Equation (\ref{eq:mp_allocation_upd}) the transmission alignment allocation can be obtained as
\begin{equation} \label{eq:mp_X_finally}
{x_{k}^{(n,l)}}^* = 
\begin{cases} 
1 & \text{if } \tau_{k}^{(n,l)} > 0 \text{ and } I^{(n)} < I_{\mathrm{max}}^{(n)}\\
0 & \text{otherwise.} \\
 \end{cases}
\end{equation}

\subsection{Algorithm Development}

In line with our discussions and from the expressions derived in Section \ref{subsec:mp_effective}, the 
MP-based resource allocation approach is outlined in \textbf{Algorithm \ref{alg:mp_rec_alloc}}. The underlay transmitters and the resources (e.g., MBS)  exchange the messages in an iterative manner. The MBS assigns the resource to the transmitters considering the node marginals, as well as the interference experienced on the RBs. The algorithm terminates when the sum data rate is reached to a steady value, i.e., the allocation vector $\mathbf{X}$  remains the same in successive iterations.

\begin{algorithm} [!t]
\caption{Resource allocation using message passing}
\label{alg:mp_rec_alloc}
\begin{algorithmic}[1]   
\AtBeginEnvironment{algorithmic}{\small} 
\renewcommand{\algorithmicforall}{\textbf{for each}}
\renewcommand{\algorithmiccomment}[1]{\textit{/* #1 */}}

\renewcommand{\algorithmicensure}{\textbf{Initialization:}}
\ENSURE

\STATE Estimate the CSI parameters from previous time slot.
\STATE Each underlay transmitter $k \in \mathcal{K}^{\mathrm T}$ selects a transmission alignment randomly and reports to MBS.
\STATE Initialize $t:= 1, ~\psi_{ k \rightarrow \lbrace n,l \rbrace }(0) := 0, ~\psi_{\lbrace n,l \rbrace \rightarrow k}(0) := 0$ for $\forall k, n, l$.

\renewcommand{\algorithmicensure}{\textbf{Update:}}
\vspace*{0.5em}
\ENSURE



\WHILE{$\mathbf{X}(t) \neq \mathbf{X}(t-1)$ \AND $t$ less than some predefined threshold $T_{\mathrm{max}}$}

\STATE Each underlay transmitter $k \in \mathcal{K}^{\mathrm T}$ sends the message 
\[\resizebox{0.93\textwidth}{!}{ $\psi_{ k \rightarrow \lbrace n,l \rbrace }(t) = \mathfrak{U}_{k}^{(n,l)}(t-1) - \omega \left\langle \mathfrak{U}_{k}^{(n^\prime,l^\prime)}(t-1) + \psi_{\lbrace n^\prime,l^\prime \rbrace \rightarrow k}(t-1) \right\rangle_{\sim \lbrace n,l \rbrace} - (1-\omega) \left( \mathfrak{U}_{k}^{(n,l)}(t-1) + \psi_{\lbrace n,l \rbrace \rightarrow k }(t-1) \right)$ }\]
for $\forall \lbrace n, l \rbrace \in \mathcal{N}\times \mathcal{L}$ to the MBS.

\STATE For all the resource $\forall \lbrace n, l\rbrace \in \mathcal{N} \times \mathcal{L}$, MBS sends messages  $$\psi_{\lbrace n,l \rbrace \rightarrow k}(t) = - \omega \underset{k^\prime \in \mathcal{K}^{\mathrm T}, k^\prime \neq k}{\operatorname{max}} \psi_{ k^\prime \rightarrow \lbrace n,l \rbrace }(t-1) - (1-\omega)~ \psi_{ k \rightarrow \lbrace n,l \rbrace } (t-1) $$ to each underlay transmitter $k \in \mathcal{K}^{\mathrm T}$.

\STATE Each underlay transmitter $k \in \mathcal{K}^{\mathrm T}$ computes the marginals as $\tau_{k}^{(n,l)}(t) = \psi_{ k \rightarrow \lbrace n,l \rbrace }(t) + \psi_{\lbrace n,l \rbrace \rightarrow k}(t) $ for $\forall \lbrace n, l \rbrace \in \mathcal{N} \times \mathcal{L}$ and reports to the MBS. 


\vspace{2pt}
{\footnotesize \textit{/* MBS calculates the allocation vector according to Equation (\ref{eq:mp_X_finally}) */} }
\vspace{2pt}

\STATE Set $x_{k}^{(n,l)} := 0$ for $\forall k, n, l$  \hspace{1em} \COMMENT{\footnotesize Initialize the variable to obtain final allocation}

\FORALL{$k \in \mathcal{K}^{\mathrm T} \text{ and } \lbrace n, l \rbrace \in \mathcal{N} \times \mathcal{L}$} 
\IF{$\tau_{k}^{(n,l)}(t) > 0 $} 

\STATE Set $x_{k}^{(n,l)} := 1$. ~~\COMMENT{\footnotesize Assign the resource to the transmitter} 


\STATE $\mathfrak{I}^{(n)} :=  \sum\limits_{\substack{ k^\prime =1}}^{K}\sum\limits_{l^\prime = 1}^{L}x_{k^\prime}^{(n, l^\prime)} g_{k^\prime,m_{k^\prime}^*}^{(n)} p_{k^\prime}^{(n)}$. ~\COMMENT{\footnotesize Calculate interference in RB $n$}

\IF{ $\mathfrak{I}^{(n)} \geq I_{\mathrm{max}}^{(n)}$ }

\REPEAT 

\STATE $\lbrace \hat{k}, \hat{l} \rbrace := \!\!\! \underset{k^\prime \in \mathcal{K}^{\mathrm T}, l^\prime \in \mathcal{L}}{\operatorname{argmax}}  x_{k^\prime}^{(n, l^\prime)} g_{k^\prime,m_{k^\prime}^*}^{(n)} p_{k^\prime}^{(n)}$ \COMMENT{\footnotesize Most interfering transmitter $\hat{k}$ with $p_{\hat{k}}^{(n)} = \hat{l}$ } 

\STATE Set $x_{\hat{k}}^{(n,\hat{l})} := 0$.  ~~~\COMMENT{\footnotesize Unassigned due to interference threshold violation} 

\STATE $\mathfrak{I}^{(n)} :=  \sum\limits_{\substack{ k^\prime =1}}^{K}\sum\limits_{l^\prime = 1}^{L}x_{k^\prime}^{(n, l^\prime)} g_{k^\prime,m_{k^\prime}^*}^{(n)} p_{k^\prime}^{(n)}$. ~~~\COMMENT{\footnotesize Update interference level} 

\UNTIL{$\mathfrak{I}^{(n)} < I_{\mathrm{max}}^{(n)}$}

\ENDIF

\ENDIF
\ENDFOR

\STATE MBS calculates the transmission alignment allocation vector $\mathbf{X}(t) = \left[ x_{k}^{(n,l)} \right]_{\forall k, n, l}$ for the iteration $t$.



\STATE Update $t:= t + 1$.

\ENDWHILE

\renewcommand{\algorithmicensure}{\textbf{Allocation:}}
\vspace*{0.5em}
\ENSURE

\STATE  Allocate the transmission alignments (e.g., RB and power levels) to the SBSs and D2D transmitters. 

\end{algorithmic}
\end{algorithm}

\subsection{Convergence, Optimality, and Complexity of the Solution}

The convergence, optimality, and complexity of the message passing approach is analyzed in the following subsections.

\subsubsection{Convergence and Optimality}

As presented in the following theorem, the message passing algorithm converges to fixed messages  within fixed number of iterations.

\begin{theorem}
The marginals and the allocation in \textbf{Algorithm \ref{alg:mp_rec_alloc}} converge to a fixed point.
\end{theorem}
\begin{proof}
The proof is constructed by utilizing the concept of \textit{contraction mapping} \cite[Chapter 3]{mp_converge}. Let the vector $\boldsymbol{\psi}(t) = \left[\psi_{ 1 \rightarrow \lbrace 1,1 \rbrace }(t), \cdots,  \psi_{ k \rightarrow \lbrace n,l \rbrace }(t), \cdots \psi_{ K \rightarrow \lbrace N,L \rbrace }(t) \right]^{\mathrm T}$ represent all the messages exchanged between the transmitters and the resources (e.g., MBS) at  iteration $t$. Let us consider the messages are translated into the mapping $\boldsymbol{\psi}(t+1) = \mathbb{T}\left( \boldsymbol{\psi}(t) \right) = \left[ \mathbb{T}_{1}^{(1,1)}\left( \boldsymbol{\psi}(t) \right), \cdots, \mathbb{T}_{K}^{(N,L)}\left( \boldsymbol{\psi}(t) \right) \right]^{\mathrm{T}}$. From the Equations (\ref{eq:mp_norm_msg_w1}) and (\ref{eq:mp_norm_msg_w2}) we can obtain $\psi_{ k \rightarrow \lbrace n,l \rbrace }(t+1) = \mathbb{T}_{k}^{(n,l)}\left( \boldsymbol{\psi}(t) \right)$ as follows:
\begin{align}
\mathbb{T}_{k}^{(n,l)}\left( \boldsymbol{\psi}(t) \right) = \omega \left( \mathfrak{U}_{k}^{(n,l)}(t) -  \mathfrak{U}_{k}^{(n^\prime,l^\prime)}(t)\right)~ + \nonumber \\ 
\omega \left( \omega \underset{k^\prime \in \mathcal{K}^{\mathrm T}, k^\prime \neq k}{\operatorname{max}} \psi_{ k^\prime \rightarrow \lbrace n^\prime,l^\prime \rbrace }(t) + (1-\omega) \psi_{ k \rightarrow \lbrace n^\prime,l^\prime \rbrace }(t)  \right)~ + \nonumber \\ 
(1- \omega) \left( \omega \underset{k^\prime \in \mathcal{K}^{\mathrm T}, k^\prime \neq k}{\operatorname{max}} \psi_{ k^\prime \rightarrow \lbrace n,l \rbrace }(t) + (1-\omega) \psi_{ k \rightarrow \lbrace n,l \rbrace }(t)   \right).
\end{align}
For any vector $\mathbf{u}$ and $\mathbf{v}$, any generic mapping $\mathbb{T}$ is a contraction if ${\parallel \mathbb{T} (\mathbf{u}) - \mathbb{T}( \mathbf{v}) \parallel}_{\infty} \leq \varepsilon {\parallel \mathbf{u} - \mathbf{v} \parallel}_{\infty}$, where $\varepsilon < 1$ is the modulus of the mapping \cite[Chapter 3]{mp_converge}. From \cite{remp_proof}, it can be shown that the mapping $\mathbb{T} : \mathbb{R}^{KNF} \rightarrow \mathbb{R}^{KNF}$ is a contraction under the maximum norm, e.g., ${\parallel \mathbb{T}\left( \boldsymbol{\psi} \right) \parallel}_{\infty} = \underset{k \in \mathcal{K}^{\mathrm T}, n \in \mathcal{N}, l \in \mathcal{L}}{\operatorname{max}} |\mathbb{T}_{k}^{(n,l)}\left( \boldsymbol{\psi} \right)|$. Since the contraction mappings have a unique fixed point convergence property for any initial vector, the proof concludes with that fact that message passing algorithm converges to a fixed marginal and hence to a fixed allocation vector $\mathbf{X}$.
\end{proof}

The following theorem presents the fixed convergence point of the message passing algorithm is an optimal solution of the original resource allocation problem.

\begin{theorem}
The allocation obtained by message passing algorithm converges to the optimal solution of resource allocation problem $\mathbf{P\ref{opt:combopt}}$.
\end{theorem}
\begin{proof}
The theorem is proved by contradiction. Let us consider that the solution $\widetilde{\mathbf{X}}$ obtained by message passing algorithm is  not optimal and let $\mathbf{X}^*$ be the optimal solution obtained by solving $\mathbf{P\ref{opt:combopt}}$. Let us further assume that there are $\chi \leq |\mathbf{X}|$ entries (e.g., allocations) that differ between $\widetilde{\mathbf{X}}$ and  $\mathbf{X}^*$. In addition, let $\widetilde{\mathcal{N}} \times \widetilde{\mathcal{L}} \subseteq \mathcal{N} \times \mathcal{L}$ denote the subset of resources for which two allocations differ. For each $\lbrace \tilde{n}, \tilde{l} \rbrace \in \widetilde{\mathcal{N}} \times \widetilde{\mathcal{L}}$ there is a transmitter $\kappa_{\lbrace \tilde{n}, \tilde{l} \rbrace}$ such that $\tilde{x}_{\kappa_{\lbrace \tilde{n}, \tilde{l} \rbrace}}^{(\tilde{n}, \tilde{l})} = 1$ and $x_{\kappa_{\lbrace \tilde{n}, \tilde{l} \rbrace}}^{*(\tilde{n}, \tilde{l})} = 0$, and a transmitter $\ddot{\kappa}_{\lbrace \tilde{n}, \tilde{l} \rbrace} \neq \kappa_{\lbrace \tilde{n}, \tilde{l} \rbrace}$ such that $\tilde{x}_{\ddot{\kappa}_{\lbrace \tilde{n}, \tilde{l} \rbrace}}^{(\tilde{n}, \tilde{l})} = 0$ and $x_{\ddot{\kappa}_{\lbrace \tilde{n}, \tilde{l} \rbrace}}^{*(\tilde{n}, \tilde{l})} = 1$. Hence, the assignment of resource  $\lbrace \tilde{n}, \tilde{l} \rbrace $ to transmitter $\kappa_{\lbrace \tilde{n}, \tilde{l} \rbrace}$ implies that the marginal $\tau_{\ddot{\kappa}_{\lbrace \tilde{n}, \tilde{l} \rbrace}}^{(\tilde{n}, \tilde{l})} < 0$ and the following set of inequalities hold for each $\lbrace \tilde{n}, \tilde{l} \rbrace \in \widetilde{\mathcal{N}} \times \widetilde{\mathcal{L}}$:
\begin{align}
\tau_{\ddot{\kappa}_{\lbrace \tilde{n}, \tilde{l} \rbrace}}^{(\tilde{n}, \tilde{l})} = \omega \left[ \left( \mathfrak{U}_{\ddot{\kappa}_{\lbrace \tilde{n}, \tilde{l} \rbrace}}^{(\tilde{n}, \tilde{l})} + \psi_{\lbrace \tilde{n}, \tilde{l} \rbrace \rightarrow \ddot{\kappa}_{\lbrace \tilde{n}, \tilde{l} \rbrace} } \right) - \left( \mathfrak{U}_{\ddot{\kappa}_{\lbrace \tilde{n}, \tilde{l} \rbrace}}^{(n^\prime, l^\prime)} + \psi_{\lbrace n^\prime, l^\prime \rbrace \rightarrow \ddot{\kappa}_{\lbrace \tilde{n}, \tilde{l} \rbrace} } \right) \right] < 0
\end{align}
where $\lbrace n^\prime, l^\prime \rbrace$ is the resource as represented in Equation (\ref{eq:mp_msg_norm1}). According to our assumption, the resource  $\lbrace n^\prime, l^\prime \rbrace$ also belongs to $\widetilde{\mathcal{N}} \times \widetilde{\mathcal{L}}$. Hence, $\sum\limits_{\lbrace \tilde{n}, \tilde{l} \rbrace \in \widetilde{\mathcal{N}} \times \widetilde{\mathcal{L}}}\tau_{\ddot{\kappa}_{\lbrace \tilde{n}, \tilde{l} \rbrace}}^{(\tilde{n}, \tilde{l})} = \omega \left( \Delta \mathfrak{U} + \Delta \psi \right)$ where
\begin{align}
\Delta \mathfrak{U} &= \sum\limits_{\lbrace \tilde{n}, \tilde{l} \rbrace \in \widetilde{\mathcal{N}} \times \widetilde{\mathcal{L}}} \left( \mathfrak{U}_{\ddot{\kappa}_{\lbrace \tilde{n}, \tilde{l} \rbrace}}^{(\tilde{n}, \tilde{l})} - \mathfrak{U}_{\kappa_{\lbrace \tilde{n}, \tilde{l} \rbrace}}^{(\tilde{n}, \tilde{l})} \right) \nonumber \\
&= \sum_{k=1}^{K}\sum_{N=1}^{N}\sum_{l=1}^{L} x_{k}^{*(n,l)} \mathfrak{U}_{k}^{(n,l)} - \sum_{k=1}^{K}\sum_{N=1}^{N}\sum_{l=1}^{L} \tilde{x}_{k}^{(n,l)} \mathfrak{U}_{k}^{(n,l)}
\end{align}
and $\Delta \psi = \sum\limits_{\lbrace \tilde{n}, \tilde{l} \rbrace \in \widetilde{\mathcal{N}} \times \widetilde{\mathcal{L}}} \left( \psi_{\lbrace \tilde{n}, \tilde{l}\rbrace \rightarrow \ddot{\kappa}_{\lbrace \tilde{n}, \tilde{l} \rbrace}} - \psi_{\lbrace \tilde{n}, \tilde{l}\rbrace \rightarrow \kappa_{\lbrace \tilde{n}, \tilde{l} \rbrace}} \right)$. After some algebraic manipulations (for details refer to \cite{remp_proof}) we can obtain $\frac{2 (1- \omega)}{\omega} \!\!\!\! \sum\limits_{\lbrace \tilde{n}, \tilde{l} \rbrace \in \widetilde{\mathcal{N}} \times \widetilde{\mathcal{L}}} \!\!\! \tau_{\ddot{\kappa}_{\lbrace \tilde{n}, \tilde{l} \rbrace}}^{(\tilde{n}, \tilde{l})} \leq  \Delta \mathfrak{U}$. Since $0 < \omega < 1$ and both the variables $\sum\limits_{\lbrace \tilde{n}, \tilde{l} \rbrace \in \widetilde{\mathcal{N}} \times \widetilde{\mathcal{L}}} \tau_{\ddot{\kappa}_{\lbrace \tilde{n}, \tilde{l} \rbrace}}^{(\tilde{n}, \tilde{l})}$ and $ \Delta \mathfrak{U}$ are positive,  our assumption that $\widetilde{\mathbf{X}}$ is not optimal is contradicted and the proof follows.
\end{proof}

\subsubsection{Complexity}

If the message passing algorithm requires $T < T_{\mathrm{max}}$ iterations to converge, it is straightforward to verify that the time complexity at each MBS is of $\mathcal{O}\left( T K N L \right)$. Similarly, considering a standard sorting algorithm that outputs the term $\left\langle \mathfrak{U}_{k}^{(n^\prime,l^\prime)} + \psi_{\lbrace n^\prime,l^\prime \rbrace \rightarrow k} \right\rangle_{\sim \lbrace n,l \rbrace}$ in order to generate the message $\psi_{ k \rightarrow \lbrace n,l \rbrace }$ with worst-case complexity of $\mathcal{O}\left( NL \log \left( NL \right) \right)$, the overall time complexity at each underlay transmitter is of $\mathcal{O} \left(T {(NL)}^2 \log \left( NL \right) \right)$.

\section{Auction-Based Resource Allocation} \label{sec:am_ra}

Our final solution approach for the resource allocation is the distributed auction algorithm. The allocation using auction is based on the \textit{bidding} procedure, where the agents (i.e., underlay transmitters) bid for the resources (e.g., RB and power level). The transmitters select the bid for the resources based on the \textit{costs} (e.g., the interference caused to the MUEs) of using the resource. The desired assignment relies on the appropriate selection of the bids. 
The unassigned transmitters raise the cost of using resource and bid for the resources simultaneously. Once the bids from all the transmitters are available, the resources are assigned to the highest bidder. An overview of auction approach 
is presented in the following.

\subsection{Overview of the Auction Approach} \label{subsec:auc_overview}

In a generic auction-based assignment model, every resource $j$ associated with a cost $c_j$ and each agent $i$ can get the benefit $B_{ij}$ from the resource $j$. Given the benefit $B_{ij}$, every agent $i$ who wishes to be assigned with the resource $j$, needs to pay the cost $c_j$. The net value (e.g., utility) that an agent $i$ can get from the resource $j$ is given by $B_{ij} - c_j$. The auction procedure involves the assignment of agent $i$ with the resource $j^\prime$ which provides the maximal net value, i.e.,
\begin{equation} \label{eq:auction_price}
B_{ij^\prime} - c_{j^\prime} = \underset{j}{\operatorname{max}} \left\lbrace B_{ij} - c_{j} \right\rbrace.
\end{equation}
If the condition given in Equation (\ref{eq:auction_price}) is satisfied for all the agents $i$, the assignment and the set of costs are referred to as \textit{equilibrium} \cite{auction_org}. However, in many practical problems, obtaining an equilibrium assignment is not straightforward due to the possibility of cycles. In particular, there may be cases where the agents contend for a small number of equally desirable resources without increasing the cost, which creates cycle (e.g., infinite loop) in the auction process. To avoid this difficulty, the notion of \textit{almost equilibrium} is introduced in the literature. The assignment and the set of costs are said to be almost equilibrium when the net value for assigning each agent $i$ with the resource $j^\prime$ is within a constant $\epsilon >0$ of being maximal. Hence, in order to be an almost equilibrium assignment, the following condition needs to be satisfied for all the agents \cite{auction_org}:
\begin{equation} \label{eq:auction_price_comslac}
B_{ij^\prime} - c_{j^\prime} \geq \underset{j}{\operatorname{max}} \left\lbrace B_{ij} - c_{j} \right\rbrace - \epsilon.
\end{equation}
The condition in Equation (\ref{eq:auction_price_comslac}) is known as $\epsilon$\textit{-complementary slackness}. When $\epsilon = 0$, Equation (\ref{eq:auction_price_comslac}) reduces to ordinary complementary slackness given by Equation (\ref{eq:auction_price}). 

For instance, let the variable $\Theta_i = j$ denote that agent $i$ is assigned with the resource $j$. In addition, let $c_{ij}$ denote the cost that agent $i$ incurs in order to be assigned with resource $j$ and $\mathfrak{b}_{ij}$ is the bidding information (i.e., highest bidder) available to the agent $i$ about resource $j$. The auction procedure evolves in an iterative manner. Given the the assignment $\Theta_i$, the set of costs $\left[c_{ij}\right]_{\forall ij}$, and the set of largest bidders $\left[\mathfrak{b}_{ij}\right]_{\forall ij}$ of previous iteration, the agents locally update the costs and the highest bidders for current iteration. In particular, the costs $c_{ij}(t)$ and bidding information $\mathfrak{b}_{ij}(t)$ available to the agent $i$ about resource $j$  for iteration $t$ are updated from the previous iteration as follows \cite{auction_base}:
\begin{align}
c_{ij}(t) &= \underset{i^\prime, i^\prime \neq i}{\operatorname{max}} \left\lbrace c_{ij}(t-1), c_{i^\prime j}(t-1) \right\rbrace \label{eq:auc_cost}\\
\mathfrak{b}_{ij}(t) &= \underset{i^* \in \underset{i^\prime, i^\prime \neq i}{\operatorname{~argmax}} \left\lbrace c_{ij}(t-1), c_{i^\prime j}(t-1) \right\rbrace }{\operatorname{max}} \left\lbrace \mathfrak{b}_{i^*j}(t-1) \right\rbrace. \label{eq:auc_bid}
\end{align} 
The above update equations ensure that the agents will have the updated maximum cost of the resource $j$ (i.e., $c_j \triangleq \underset{i}{\operatorname{max}} \lbrace c_{i j} \rbrace$) and the corresponding highest bidder for that resource. Once the update cost and bidding information are available, agent $i$ checks whether the cost of the resource currently assigned to  agent $i$, e.g., $c_{i \Theta_i(t-1)}$ has been increased by any other agents. If so, the current assignment obtained from previous iteration may not be at (almost) equilibrium and the agent needs to select a new assignment, e.g., $\Theta_{i}(t) = \underset{j}{\operatorname{argmax}} \left\lbrace B_{ij}(t) - c_{ij} (t) \right\rbrace$. In order to update the cost for new assignment (e.g., $\Theta_{i}(t)$) for any iteration $t$, the agent will use the following cost update rule \cite{auction_base}: 
\begin{equation} \label{eq:auc_costupd}
c_{ij}(t) = c_{ij}(t-1) + \Delta_i(t-1)
\end{equation}
where $\Delta_i$ is given by
\begin{equation} \label{eq:auc_price_uptate}
\Delta_i(t-1) = \underset{j}{\operatorname{max}} \left\lbrace B_{ij}(t-1) - c_{ij}(t-1)  \right\rbrace - \underset{j^\prime \neq \Theta_i(t)}{\operatorname{max}} \left\lbrace B_{ij^\prime}(t-1) - c_{ij^\prime}(t-1)  \right\rbrace + \epsilon.
\end{equation}
The variable $\underset{j}{\operatorname{max}} \left\lbrace B_{ij}(t-1) - c_{ij}(t-1)  \right\rbrace$ and $\underset{j^\prime \neq \Theta_i(t)}{\operatorname{max}} \left\lbrace B_{ij^\prime}(t-1) - c_{ij^\prime}(t-1)  \right\rbrace$ denote the maximum and second maximum net utility, respectively.  Note that $\Delta_i$ is always greater than zero as $\epsilon >0$ and by definition $\underset{j}{\operatorname{max}} \left\lbrace B_{ij}(t-1) - c_{ij}(t-1)  \right\rbrace > \underset{j^\prime \neq \Theta_i(t)}{\operatorname{max}} \left\lbrace B_{ij^\prime}(t-1) - c_{ij^\prime}(t-1)  \right\rbrace$. Since $c_{i \Theta_i(t)}(t)$ is the highest cost for iteration $t$, agent $i$ can also update the bidding information, e.g., $\mathfrak{b}_{i \Theta_i(t)}(t) = i$. Accordingly, the cost update rule using $\Delta_i$ as given in Equation (\ref{eq:auc_costupd}) ensures that the assignment and the set of costs are almost at equilibrium \cite{auction_base}.

\subsection{Auction for Radio Resource Allocation}

Based on the discussion provided in the preceding section, in the following, we present the auction-based resource allocation scheme. We present the cost model and use the concept of auction to develop the resource allocation algorithm in our considered heterogeneous network setup.

\subsubsection{Cost Function}

Let us consider the utility function given by Equation (\ref{eq:sm_utility}). Recall that the term $w_2 \left( I^{(n)} - I_{\mathrm{max}}^{(n)} \right) $ in Equation (\ref{eq:sm_utility}) represents the cost (e.g., interference caused by underlay transmitters to the MUE) of using the RB $n$. In particular, when the transmitter $k$ is transmitting with power level $l$, the cost of using RB $n$ can be represented by
\begin{align} \label{eq:auc_cost_ra}
c_{k}^{(n,l)} &=  w_2 \left( I^{(n)} - I_{\mathrm{max}}^{(n)} \right) = w_2 \left( \sum\limits_{k^\prime =1}^{K}\sum\limits_{l^\prime = 1}^{L}x_{k^\prime}^{(n, l^\prime)} g_{k^\prime,m_{k^\prime}^*}^{(n)} p_{k^\prime}^{(n)} - I_{\mathrm{max}}^{(n)} \right) \nonumber \\
&= w_2 \left(  g_{k,m_k^*}^{(n)} l +   \sum\limits_{\substack{ k^\prime \in \mathcal{K}^{\mathrm{T}}, k^\prime \neq k}}\sum\limits_{l^\prime = 1}^{L}x_{k^\prime}^{(n, l^\prime)} g_{k^\prime,m_{k^\prime}^*}^{(n)} p_{k^\prime}^{(n)} - I_{\mathrm{max}}^{(n)} \right).
\end{align}
Let the parameter $C_{k}^{(n,l)} = \max \lbrace 0, c_{k}^{(n,l)} \rbrace$ and accordingly the cost $C_{k}^{(n,l)} = 0$ only if $I^{(n)} \leq I_{\mathrm{max}}^{(n)}$.  Notice that using the cost term we can represent Equation (\ref{eq:sm_utility}) as $$\mathfrak{U}_{k}^{(n,l)} = w_1 \mathscr{R}\left(\Gamma_{u_k}^{(n, l)}\right)  -  w_2 \left( I^{(n)} - I_{\mathrm{max}}^{(n)} \right) = B_{k}^{(n,l)} - c_{k}^{(n,l)} = B_{k}^{(n,l)} - C_{k}^{(n,l)}$$ where $B_{k}^{(n,l)} = w_1 \mathscr{R}\left(\Gamma_{u_k}^{(n, l)}\right)$, and $c_{k}^{(n,l)}$ is given by Equation (\ref{eq:auc_cost_ra}). The variable $B_{k}^{(n,l)}$ is proportional to the data rate achieved by transmitter $k$ using resource $\lbrace n,l \rbrace$. Analogous to the discussion of previous section, $\mathfrak{U}_{k}^{(n,l)}$ represents the net benefit that transmitter $k$ obtains from the resource $\lbrace n,l\rbrace$.

Let $\mathfrak{b}_{k}^{( n,l)}$ denote the local bidding information available to transmitter $k$ for the resource $\lbrace n,l \rbrace$. For notational convenience, let us assume that $\Theta : [k]_{k = 1, \cdots, K} \rightarrow \left[ \lbrace n, l  \rbrace \right]_{\substack{n = 1, \cdots, N \\ l = 1, \cdots, L}}$ denotes the mapping between the transmitters and the resources, i.e., $\Theta_k = \lbrace n,l \rbrace$ represents the assignment of resource $\lbrace n,l \rbrace$ to  transmitter $k$. Hence we represent by $C_{k}^{\Theta_k}$ the cost of using the resource $\lbrace n,l \rbrace$ obtained by the assignment $\Theta_k = \lbrace n,l \rbrace$. Similarly, given $\Theta_k = \lbrace n,l \rbrace$ the variable $\mathfrak{b}_{k}^{\Theta_k} \equiv \mathfrak{b}_{k}^{( n,l)}$ denotes the local bidding information about the resource $\lbrace n,l \rbrace$ available to the transmitter $k$. Note that $\Theta_k = \lbrace n,l \rbrace$ also implies $x_{k}^{(n,l)} = 1$. In other words, $\Theta_k = \lbrace n,l \rbrace$ denote the non-zero entry of the vector $\mathbf{x}_k = \left[ x_{k}^{(n,l)} \right]_{\forall n,l}$. Since each underlay transmitter $k$ selects only one resource $\lbrace n, l\rbrace$, only a single entry in the vector $\mathbf{x}_k$ is non-zero.

\subsubsection{Update of Cost and Bidder Information}

In order to obtain the updated cost and bidding information, we utilize similar concept given by Equations (\ref{eq:auc_cost})-(\ref{eq:auc_price_uptate}). At the beginning of the auction procedure, each underlay transmitter updates the cost as $C_{k}^{(n,l)}(t) = \underset{k^\prime \in \mathcal{K}^{\mathrm T}, k^\prime \neq k}{\operatorname{max}} \left\lbrace C_{k}^{(n,l)}(t-1), C_{k^\prime}^{(n,l)}(t-1) \right\rbrace $. In addition, as described by Equation (\ref{eq:auc_bid}), the information of maximum bidder is obtained by $\mathfrak{b}_{k}^{( n,l)}(t) = \mathfrak{b}_{k^*}^{( n,l)}(t-1)$ where $k^* = \underset{k^\prime \in \mathcal{K}^{\mathrm T}, k^\prime \neq k}{\operatorname{argmax}} \left\lbrace C_{k}^{(n,l)}(t-1), C_{k^\prime}^{(n,l)}(t-1) \right\rbrace$. When the transmitter $k$ needs to select a new assignment, i.e., $\Theta_k = \lbrace \hat{n},\hat{l}\rbrace$, the transmitter increases the cost of using the resource, e.g., $C_{k}^{(\hat{n}, \hat{l})}(t) = C_{k}^{(\hat{n}, \hat{l})}(t-1) +\Delta_k(t-1)$, and $\Delta_k(t-1)$ is given by
\begin{align} \label{eq:auc_costupdate}
\Delta_k(t-1) = \underset{\lbrace n^\prime, l^\prime \rbrace \in \mathcal{N}\times\mathcal{L}}{\operatorname{max}} \mathfrak{U}_{k}^{(n^\prime,l^\prime)}(t-1) - \underset{\substack{\lbrace n^\prime, l^\prime \rbrace \in \mathcal{N}\times\mathcal{L} \\ n^\prime \neq \hat{n}, l^\prime \neq \hat{l} }}{\operatorname{max}} \mathfrak{U}_{k}^{(n^\prime,l^\prime)}(t-1) + \epsilon
\end{align}
where $\epsilon > 0$ indicates the minimum bid requirement parameter. Similar to Equation (\ref{eq:auc_price_uptate}), the term $\underset{\lbrace n^\prime, l^\prime \rbrace \in \mathcal{N}\times\mathcal{L}}{\operatorname{max}} \mathfrak{U}_{k}^{(n^\prime,l^\prime)}(t-1) - \underset{\substack{\lbrace n^\prime, l^\prime \rbrace \in \mathcal{N}\times\mathcal{L} \\ n^\prime \neq \hat{n}, l^\prime \neq \hat{l} }}{\operatorname{max}} \mathfrak{U}_{k}^{(n^\prime,l^\prime)}(t-1)$ denotes the difference between the maximum and the second to the maximum utility value. In the case when the transmitter $k$ does not prefer to be assigned with a new resource, the allocation from the previous iteration will remain unchanged, i.e., $\Theta_k(t) = \Theta_k(t-1)$, and consequently, $\mathbf{x}_k(t) = \mathbf{x}_k(t-1)$.

\begin{algorithm} [!t]
\AtBeginEnvironment{algorithmic}{\small} 
\caption{Auction method for any underlay transmitter $k$}
\label{alg:auc_loc}
\begin{algorithmic}[1]   
\renewcommand{\algorithmicrequire}{\textbf{Input:}}

\renewcommand{\algorithmicensure}{\textbf{Output:}} 
\renewcommand{\algorithmicforall}{\textbf{for each}}
\renewcommand{\algorithmiccomment}[1]{\textit{/* #1 */}}

\REQUIRE Parameters from previous iteration: an assignment $\mathbf{X}(t-1) = \left[ \mathbf{x}_1(t-1), \cdots \mathbf{x}_K(t-1)  \right]^{\mathsf{T}}$, aggregated interference $I^{(n)}(t-1)$ for $\forall n$, cost of using resources $\mathbf{C}(t-1) = \left[ C_{k}^{(n,l)}(t-1)\right]_{\forall k,n, l}$ and the highest bidders of the resources $\mathfrak{B}(t-1) = \left[ \mathfrak{B}_k(t) \right]_{\forall k}$ where $\mathfrak{B}_k(t) = \left[\mathfrak{b}_{k}^{( n,l)}(t) \right]_{\forall n, l}$.

\ENSURE The allocation variable $\mathbf{x}_k(t) = \left[x_{k}^{(n,l)}\right]_{\forall n, l}$, updated costs $\mathbf{C}_k(t) = \left[ C_{k}^{(n,l)}(t)\right]_{\forall n, l}$, and bidding information $\mathfrak{B}_k(t) = \left[\mathfrak{b}_{k}^{( n,l)}(t) \right]_{\forall n, l}$ at current iteration $t$ for the transmitter $k$.

\STATE Initialize $\mathbf{x}_k(t) := \mathbf{0}$.

\STATE For all the resources $\lbrace n, l\rbrace \in \mathcal{N}\times \mathcal{L}$, 
\begin{itemize}

\item Obtain the transmitter $k^* := \underset{k^\prime \in \mathcal{K}^{\mathrm T}, k^\prime \neq k}{\operatorname{argmax}} \left\lbrace C_{k}^{(n,l)}(t-1), C_{k^\prime}^{(n,l)}(t-1) \right\rbrace$ and update the highest bidder as $\mathfrak{b}_{k}^{( n,l)}(t) := \mathfrak{b}_{k^*}^{( n,l)}(t-1)$.

\item Update the cost as
$C_{k}^{(n,l)}(t) := \underset{k^\prime \in \mathcal{K}^{\mathrm T}, k^\prime \neq k}{\operatorname{max}} \left\lbrace C_{k}^{(n,l)}(t-1), C_{k^\prime}^{(n,l)}(t-1) \right\rbrace $.

\end{itemize}


\vspace*{2pt} 

{\footnotesize \textit{/* Let $\Theta_k(t-1)$ denote the assignment of transmitter $k$ at previous iteration $t-1$, i.e., $\Theta_k(t-1)$ represents the non-zero entry in the vector $\mathbf{x}_k(t-1)$. Since each transmitter uses only one transmission alignment, only a single entry in the vector $\mathbf{x}_k(t-1)$ is non-zero. When cost is greater than previous iteration and the transmitter $k$ is not the highest bidder, update the assignment */} }

\vspace*{5pt}

\IF{ $C_{k}^{\Theta_k(t-1)} (t) \geq C_{k}^{\Theta_k(t-1)} (t-1)  $ \AND $\mathfrak{b}_{k}^{\Theta_k(t-1)}(t) \neq k $ }

\STATE $\lbrace \hat{n}, \hat{l} \rbrace := \underset{\lbrace n^\prime, l^\prime \rbrace \in \mathcal{N} \times \mathcal{L}}{\operatorname{argmax}} \mathfrak{U}_{k}^{(n^\prime,l^\prime)}(t)$.  ~\COMMENT{\footnotesize Obtain the best resource for transmitter $k$} 


\STATE $\mathfrak{I}^{(\hat{n})} := g_{k,m_k^*}^{(\hat{n})} \hat{l} +   I^{(\hat{n})}$. ~\COMMENT{\footnotesize Calculate additional interference caused by transmitter $k$ for using RB $\hat{n}$}

\IF{ $\mathfrak{I}^{(\hat{n})} < I_{\mathrm{max}}^{(\hat{n})}$ }

\STATE Set $x_{k}^{(\hat{n},\hat{l})} := 1$. ~~~~~\COMMENT{\footnotesize e.g., $\Theta_{k}(t) = \lbrace \hat{n},\hat{l} \rbrace $}

\STATE Update the highest bidder for the resource $\lbrace \hat{n}, \hat{l} \rbrace $ as $\mathfrak{b}_{k}^{(\hat{n}, \hat{l}) }(t) := k$.

\STATE Increase the cost for the resource $\lbrace \hat{n}, \hat{l} \rbrace $ as $C_{k}^{(\hat{n}, \hat{l})}(t) = C_{k}^{(\hat{n}, \hat{l})}(t-1) +\Delta_k(t-1)$ where $\Delta_k(t-1)$ is given by Equation (\ref{eq:auc_costupdate}).



\ELSE 
\STATE Keep the assignment unchanged from previous iteration, i.e., $\mathbf{x}_k(t) := \mathbf{x}_k(t-1)$.

\ENDIF

\ELSE 
\STATE Keep the assignment unchanged from previous iteration, i.e., $\mathbf{x}_k(t) := \mathbf{x}_k(t-1)$.
 
\ENDIF


\end{algorithmic}
\end{algorithm}

\subsection{Algorithm Development}

\textbf{Algorithm \ref{alg:auc_alg}} outlines the auction-based resource allocation approach. Each transmitter locally executes \textbf{Algorithm \ref{alg:auc_loc}} and obtains a temporary allocation. When the execution of \textbf{Algorithm \ref{alg:auc_loc}} is finished, each underlay transmitter $k$ reports to the MBS the local information, e.g., choices for the resources, $\mathbf{x}_k = \left[ x_{k}^{(n,l)} \right]_{\forall n,l}$. Once the information (e.g., output parameters from \textbf{Algorithm \ref{alg:auc_loc}}) from all the transmitters are available to the MBS, the necessary parameters (e.g., input arguments required by \textbf{Algorithm \ref{alg:auc_loc}}) are calculated and broadcast by the MBS. \textbf{Algorithm \ref{alg:auc_loc}} repeated iteratively until the allocation variable $\mathbf{X} = \left[\mathbf{x}_k \right]_{\forall k} = \left[x_{1}^{(1, 1)}, \cdots, x_{1}^{(1, L)}, \cdots, x_{1}^{(N, L)}, \cdots, x_{K}^{(N, L)} \right]^{\mathsf{T}}$ for two successive iterations becomes similar.

\begin{algorithm} [!t]
\AtBeginEnvironment{algorithmic}{\small} 
\caption{Auction-based resource allocation}
\label{alg:auc_alg}
\begin{algorithmic}[1]   
\renewcommand{\algorithmicrequire}{\textbf{Input:}}
\renewcommand{\algorithmicensure}{\textbf{Output:}}
\renewcommand{\algorithmicforall}{\textbf{for each}}
\renewcommand{\algorithmiccomment}[1]{\textit{/* #1 */}}

\renewcommand{\algorithmicensure}{\textbf{Initialization:}}
\ENSURE

\STATE Estimate the CSI parameters from the previous time slot.

\STATE Each underlay transmitter $k \in \mathcal{K}^{\mathrm T}$ randomly selects a transmission alignment and reports to the MBS.

\STATE MBS broadcasts the assignment of all transmitters, aggregated interference of each RB, the costs and the highest bidders using pilot signals.

\STATE Initialize number of iterations $t := 1$.

\renewcommand{\algorithmicensure}{\textbf{Update:}}
\vspace*{0.5em}
\ENSURE

\WHILE{$\mathbf{X}(t) \neq \mathbf{X}(t-1)$ \AND $t$ is less than some predefined threshold $T_{\mathrm{max}}$}

\STATE Each underlay transmitter $k \in \mathcal{K}^{\mathrm T}$ locally runs the \textbf{Algorithm \ref{alg:auc_loc}} and reports the assignment $\mathbf{x}_k(t)$, the cost $\mathbf{C}_k(t)$ and the bidding information $\mathfrak{B}_k(t)$ to the MBS. 

\STATE MBS calculates the aggregated interference $I^{(n)}(t)$ for $\forall n$, the allocation variable $\mathbf{X}(t)$, information about highest bidders $\mathfrak{B}(t)$, the cost $\mathbf{C}(t)$, and broadcast to the underlay transmitters.

\STATE Update $t := t+1$.
\ENDWHILE

\renewcommand{\algorithmicensure}{\textbf{Allocation:}}
\vspace*{0.5em}
\ENSURE

\STATE Allocate the RB and power levels to the SBSs and D2D UEs.

\end{algorithmic}
\end{algorithm}

\subsection{Convergence, Complexity, and Optimality of the Auction Approach}

In the following subsections we analyze the convergence, complexity, and optimality of the solution obtained by auction algorithm.

\subsubsection{Convergence and Complexity}

For any arbitrary fixed $\epsilon >0$, the auction approach is guaranteed to converge to a fixed assignment. The following theorem shows that the auction process terminates within a fixed number of iterations.

\begin{theorem} \label{thm:auc_terminate}
The auction process terminates in a finite number of iterations.
\end{theorem}
\begin{proof}
According to our system model, each underlay transmitter selects only one transmission alignment. Hence, once each resource receives at least one bid (which implies that each transmitter is assigned to a resource), the auction process must terminate. Now if any resource $\lbrace n, l \rbrace$ receives a bid in $\hat{t}$ iterations, the cost must be greater than the initial price by $\hat{t} \epsilon$. As a result, the resource $\lbrace n, l \rbrace$ becomes costly to be assigned when compared to any resource $\lbrace n^\prime, l^\prime \rbrace$ that has not received any bid yet. The argument follows that there are two possibilities, e.g., \textit{i)} the auction process terminates in a finite iterations with each transmitter assigned to a resource, regardless of every resource receives a bid; or \textit{ii)} the auction process continues for a finite number of iterations and each resource will receive at least one bid, therefore, the algorithm terminates.
\end{proof}

At termination, the solution (e.g., allocation) obtained is almost at equilibrium, e.g., the condition in Equation (\ref{eq:auction_price_comslac}) is satisfied for all the underlay transmitters. Since the algorithm terminates after a finite number of iterations, we can show that the algorithm converges to a fixed allocation and the complexity at each transmitter is linear to the number of resources.

\begin{theorem}
The auction algorithm converges to a fixed allocation with the number of iterations of $$\mathcal{O}\left( T KNL \left \lceil {\frac{\underset{k, n, l}{\operatorname{max}} B_{k}^{(n,l)} - \underset{k, n, l}{\operatorname{min}}  B_{k}^{(n,l)}}{\epsilon} } \right\rceil \right).$$
\end{theorem}
\begin{proof}
The proof follows from the similar argument presented in \textbf{Theorem \ref{thm:auc_terminate}}. In the worst case, the total number of iterations in which a resource can receive a bid is no more than $ \Upsilon = \left \lceil {\frac{\underset{k, n, l}{\operatorname{max}} B_{k}^{(n,l)} - \underset{k, n, l}{\operatorname{min}}  B_{k}^{(n,l)}}{\epsilon} } \right\rceil$ \cite{auction_base}. Since each bid requires $\mathcal{O}\left(NL \right)$ iterations, and each iteration involves a bid by a single transmitter, the total number of iterations in \textbf{Algorithm \ref{alg:auc_alg}} is of $\mathcal{O}\left(  KNL \Upsilon \right)$. For the convergence, the allocation variable $\mathbf{X}$ needs to be unchanged for at least $T \geq 2$ consecutive iterations. Hence, the overall running time of the algorithm is $\mathcal{O}\left( T KNL \Upsilon \right)$.
\end{proof}

Note that for any transmitter node $k \in \mathcal{K}^{\mathrm T}$, the complexity of the auction process given by \textbf{Algorithm \ref{alg:auc_loc}} is linear with number of resources for each of the iterations. 

\subsubsection{Optimality}

In the following we show that the data rate obtained by the auction algorithm is within $K \epsilon$ of the maximum data rate obtained by solving the original optimization problem $\mathbf{P\ref{opt:combopt}}$.

\begin{theorem}
The data rate obtained by the distributed auction algorithm is within $K \epsilon$ of the optimal solution.
\end{theorem}
\begin{proof}
We construct the proof by using an approach similar to that presented in \cite{auction_base}. The data rate obtained by any assignment $\mathbf{X}$ will satisfy the following condition: 
\begin{equation} \label{eq:auc_prof_ineqal}
\sum_{k=1}^{K} R_{u_k} \leq \sum_{\lbrace n, l \rbrace \in \mathcal{N} \times \mathcal{L}} \widehat{C}^{(n,l)} + \sum_{k=1}^{K} \underset{ \lbrace n, l \rbrace \in \mathcal{N} \times \mathcal{L}}{\operatorname{max}} \left\lbrace B_{k}^{(n,l)} -  \widehat{C}^{(n,l)} \right\rbrace 
\end{equation}
where $\widehat{C}^{(n,l)} = \underset{k^\prime \in \mathcal{K}^{\mathrm T}}{\operatorname{max}} C_{k^\prime}^{(n,l)} $, $B_{k}^{(n,l)} = w_1 \mathscr{R}\left(\Gamma_{u_k}^{(n, l)}\right)$ and $R_{u_k}$ is given by Equation (\ref{eq:rate_ue}). The inequality given by Equation (\ref{eq:auc_prof_ineqal}) is satisfied since the first term in the right side of the inequality, e.g., $\sum\limits_{\lbrace n, l \rbrace \in \mathcal{N} \times \mathcal{L}} \widehat{C}_{(n,l)}$ 
is equal to $\sum\limits_{k=1}^{K} \sum\limits_{n=1}^{N} \sum\limits_{l=1}^{L} x_{k}^{(n,l)} C_{k}^{(n,l)}$ and the second term is not less than $\sum\limits_{k=1}^{K} \sum\limits_{n=1}^{N} \sum\limits_{l=1}^{L} x_{k}^{(n,l)}\left(  B_{k}^{(n,l)} - \widehat{C}^{(n,l)}  \right)$. Let the variable $A^* \triangleq \underset{\mathbf{X}^*}{\operatorname{max}} \sum\limits_{k=1}^{K} R_{u_k} = \sum\limits_{k=1}^{K} \sum\limits_{n = 1}^{N} \sum\limits_{l = 1}^{L}  ~{x_{k}^{(n,l)}}^{*}   B_{\mathrm {RB}} \log_2 \left(1 +  \gamma_{u_k}^{(n)} \right)$ denote the optimal achievable data rate. In addition, let the variable $D^*$ be defined as
\begin{equation}
D^* \triangleq \underset{\substack{\widehat{C}^{(n,l)} \\ \lbrace n, l \rbrace \in \mathcal{N} \times \mathcal{L}}}{\operatorname{min}} \left\lbrace \sum_{\lbrace n, l \rbrace \in \mathcal{N} \times \mathcal{L}} \widehat{C}^{(n,l)} + \sum_{k=1}^{K} \underset{ \lbrace n, l \rbrace \in \mathcal{N} \times \mathcal{L}}{\operatorname{max}} \left\lbrace B_{k}^{(n,l)} -  \widehat{C}^{(n,l)} \right\rbrace  \right\rbrace.
\end{equation}
Hence from Equation (\ref{eq:auc_prof_ineqal}), we can write $A^* \leq D^*$. Since the final assignment and the set of costs are almost at equilibrium, for any underlay transmitter $k$, the condition $\sum\limits_{n=1}^{N} \sum\limits_{l=1}^{L} x_{k}^{(n,l)}\left(  B_{k}^{(n,l)} - \widehat{C}^{(n,l)}  \right) \geq \underset{ \lbrace n, l \rbrace \in \mathcal{N} \times \mathcal{L}}{\operatorname{max}} \left\lbrace B_{k}^{(n,l)} -  \widehat{C}^{(n,l)} \right\rbrace - \epsilon$ will hold. Consequently, we can obtain the following inequality:
\begin{align}
D^* &\leq \sum_{k=1}^{K} \left( \sum_{n=1}^{N} \sum_{l=1}^{L} x_{k}^{(n,l)} \widehat{C}^{(n,l)} + \underset{ \lbrace n, l \rbrace \in \mathcal{N} \times \mathcal{L}}{\operatorname{max}} \left\lbrace B_{k}^{(n,l)} -  \widehat{C}^{(n,l)} \right\rbrace  \right) \nonumber \\
&\leq \sum_{k=1}^{K}\sum_{n=1}^{N} \sum_{l=1}^{L} x_{k}^{(n,l)}  B_{k}^{(n,l)} + K \epsilon \leq \sum_{k=1}^{K} R_{u_k} + K \epsilon \leq A^* + K \epsilon. 
\end{align}
Since $A^* \leq D^*$, the data rate achieved by the auction algorithm is within $K \epsilon$ of the optimal data rate $A^*$ and the proof follows. 
\end{proof}

\section{Qualitative Comparison Among the Resource Allocation Schemes} \label{sec:comparisons}

In this section, we compare the different resource allocation schemes discussed above based on several criteria (e.g., flow of algorithm execution, information requirement and algorithm overhead, complexity and optimality of the solution, convergence behavior etc.). We term the centralize solution (which can be obtained by solving the optimization problem $\mathbf{P\ref{opt:combopt}}$) as COS (centralized optimal scheme) and compare it with the distributed solutions. A comparison among the resource allocation schemes is presented in Table \ref{tab:comp}.


\begin{table}[!h] 
\centering
\begin{footnotesize}
\begin{tabular}{P{2.0cm} P{2.5cm} P{2.5cm} P{2.5cm} P{2.5cm} }
\toprule 
\multirow{2}{*}{Criterion} &  \multicolumn{4}{c}{Schemes}\\  \cline{2-5} 
&  \multicolumn{1}{c}{COS} & \multicolumn{1}{c}{Stable matching} & \multicolumn{1}{c}{Message passing} & \multicolumn{1}{c}{Auction method} \\
\midrule 
Type of the solution & Centralized & Distributed & Distributed & Distributed\\
\midrule 
Algorithm execution & MBS solves the resource optimization problem (e.g., $\mathbf{P\ref{opt:combopt}}$) & MBS and underlay transmitters locally update the preference profiles, MBS runs the matching subroutine & MBS and underlay transmitters alliteratively exchange the messages, MBS computes the marginals and selects allocation & Each underlay transmitters locally runs the auction subroutine, MBS collects the parameters from all the transmitters and broadcast required parameters needed for the auction subroutine\\ 
\midrule 
Optimality & Optimal & Weak Pareto optimal & Optimal subject to the weight $\omega$ & Within $K\epsilon$ to the optimal\\
\midrule 
Complexity & $\mathcal{O}\left( \left(NL \right)^{K} \right)$ at the MBS & \resizebox{.99\hsize}{!}{$\mathcal{O}\left( T NL \log(NL) \right)$} at the transmitters, $\mathcal{O}(TKNL) $ at the MBS & \resizebox{.99\hsize}{!}{  $\mathcal{O} \left(T {(NL)}^2 \log \left( NL \right) \right)$} at the transmitters, $\mathcal{O}\left( T K N L \right)$ at the MBS & For each iteration linear with $N, L$ at the transmitters, overall running time $\mathcal{O}\left( T KNL \Upsilon \right)$\\
\midrule 
Convergence behavior &  N/A & Converges to a stable matching and hence to a fixed allocation & Converges to a fixed marginal and to a fixed allocation & Converges to a fixed allocation within $K\epsilon$ of the optimal \\
\midrule 
Information required by the MBS & Channel gains (e.g., CSI parameters) between all the links of the network & The preference profiles and the channel gains $G_{k}^{(n)} = \left[ g_{k,m_{k}^*}^{(n)} \right]_{\forall k, n}$ & The messages $\left[ \psi_{k \rightarrow \lbrace n,l \rbrace} \right]_{\forall k,n,l}$ and the channel gains $G_{k}^{(n)} = \left[ g_{k,m_{k}^*}^{(n)} \right]_{\forall k, n}$  & The channel gains $G_{k}^{(n)} = \left[ g_{k,m_{k}^*}^{(n)} \right]_{\forall k, n}$, local assignments $\mathbf{x}_k$, the cost $\mathbf{C}_k$, and the bidding information $\mathfrak{B}_k$ for $\forall k$  \\
\midrule 
Algorithm overhead & High (exponential) computational complexity, requirement of all CSI parameters of the network & Build the preference profiles, exchange information to update preference profiles, execution of matching subroutine & Calculation and exchange of messages, computation of the marginals & Computation and exchange of the parameters, e.g.,  $I^{(n)}$ for $\forall n$, the allocation vector $\mathbf{X}$, information about highest bidders $\mathfrak{B}$, the cost vector $\mathbf{C}$\\
\toprule

\end{tabular}

\end{footnotesize}
\caption{Comparison among different resource allocation approaches}{}
\label{tab:comp}
\end{table}

\section{Chapter Summary and Conclusion} \label{sec:conclusion}

We have presented three comprehensive distributed solution approaches for the future 5G cellular mobile communication systems. Considering a heterogeneous multi-tier 5G network, we have developed distributed radio resource allocation algorithms using three different mathematical models (e.g., stable matching, message passing, and auction method). The properties (e.g., convergence, complexity, optimality) of these distributed solutions are also briefly analyzed. To this end, a qualitative comparison of these schemes is illustrated.  

The solution tools presented in this chapter can also be applicable to address the resource allocation problems in other enabling technologies for 5G systems. In particular, the mathematical tools presented in this chapter open up new opportunities to investigate other network models, such as resource allocation problems for wireless virtualization \cite{wnv_1} and cloud-based radio access networks \cite{c_ran1}. In such systems, these modeling tools need to be customized accordingly based on the objective and constraints required for the resource allocation problem. 

In addition to the presented solutions, there are few game theoretical models which have not been covered in this chapter. However, these game models can also be considered as potential distributed solution tools. Different from traditional cooperative and non-cooperative games, the game models (such as mean field games \cite{mfg_schedule, mfg_crn}, evolutionary games \cite{evo_wireless} etc.) are scalable by nature, and hence applicable to model such large heterogeneous 5G networks. Utilizing those advanced game models for the resource allocation problems and analyzing the performance (e.g., data rate, spectrum and energy efficiency etc.) of 5G systems could be an interesting area of research.

\clearpage
\bibliographystyle{IEEEtran}

\section*{Additional Reading}


\newcounter{enumTemp}

{\normalsize \noindent $\bullet$~ \textit{5G and Heterogeneous Networks:}}

\begin{enumerate}[label={[\roman*]}]

\item
J.~Andrews, S.~Buzzi, W.~Choi, S.~Hanly, A.~Lozano, A.~Soong, and J.~Zhang,
  ``What will 5G be?'' \emph{IEEE Journal on Selected Areas in Communications},
  vol.~32, no.~6, pp. 1065--1082, June 2014.
  
  
  \item
B.~Bangerter, S.~Talwar, R.~Arefi, and K.~Stewart, ``Networks and devices for
  the 5G era,'' \emph{IEEE Communications Magazine}, vol.~52, no.~2, pp.
  90--96, February 2014.
  
\item
E.~Hossain, M.~Rasti, H.~Tabassum, and A.~Abdelnasser, ``Evolution toward 5G
  multi-tier cellular wireless networks: an interference management
  perspective,'' \emph{IEEE Wireless Communications}, vol.~21, no.~3, pp.
  118--127, June 2014.

\item
Y.~Lee, T.~Chuah, J.~Loo, and A.~Vinel, ``{Recent advances in radio resource
  management for heterogeneous LTE/LTE-A networks},'' \emph{IEEE Communications
  Surveys and Tutorials}, vol.~PP, no.~99, pp. 1--39, 2014.

\end{enumerate}

\vspace{0.5em}
{\normalsize \noindent $\bullet$~ \textit{Stable Matching:}}

\begin{enumerate}[label={[\roman*]}, resume]

\item
K.~Iwama and S.~Miyazaki, ``A survey of the stable marriage problem and its
  variants,'' in \emph{International Conference on Informatics Education and
  Research for Knowledge-Circulating Society (ICKS)}, January 2008, pp. 131--136.

\item 
X.~Feng, G.~Sun, X.~Gan, F.~Yang, X.~Tian, X.~Wang, and M.~Guizani,
  ``Cooperative spectrum sharing in cognitive radio networks: a
   distributed
  matching approach,'' \emph{IEEE Transactions on Communications}, vol.~62,
  no.~8, pp. 2651--2664, August 2014.

\item
A.~Leshem, E.~Zehavi, and Y.~Yaffe, ``Multichannel opportunistic carrier
  sensing for stable channel access control in cognitive radio systems,''
  \emph{IEEE Journal on Selected Areas in Communications}, vol.~30, no.~1, pp.
  82--95, January 2012.

\end{enumerate}

\vspace{0.5em}
{\normalsize \noindent $\bullet$~ \textit{Message Passing:}}

\begin{enumerate}[label={[\roman*]}, resume]

\item 
M.~Hasan and E.~Hossain, ``Distributed resource allocation for relay-aided
  device-to-device communication: a message passing approach,'' \emph{IEEE
  Transactions on Wireless Communications}, 2014.
\item 
A.~Abrardo, M.~Belleschi, P.~Detti, and M.~Moretti, ``Message passing resource
  allocation for the uplink of multi-carrier multi-format systems,'' \emph{IEEE
  Transactions on Wireless Communications}, vol.~11, no.~1, pp. 130--141, 2012.
  
  \item K.~Yang, N.~Prasad, and X.~Wang, ``{A message-passing approach to distributed
  resource allocation in uplink DFT-Spread-OFDMA systems},'' \emph{IEEE
  Transactions on Communications}, vol.~59, no.~4, pp. 1099--1113, 2011.

\end{enumerate}

\vspace{0.5em}
{\normalsize \noindent $\bullet$~ \textit{Auction Algorithm:}}

\begin{enumerate}[label={[\roman*]}, resume]
\item Y.~Zhang, C.~Lee, D.~Niyato, and P.~Wang, ``Auction approaches for resource   allocation in wireless systems: a survey,'' \emph{IEEE Communications Surveys and  Tutorials}, vol.~15, no.~3, pp. 1020--1041, 2013.

\item I.~Koutsopoulos and G.~Iosifidis, ``Auction mechanisms for network resource allocation,'' in \emph{8th International Symposium on Modeling and
  Optimization in Mobile, Ad Hoc and Wireless Networks (WiOpt)}, May 2010, pp. 554--563.
  
  \item 
K.~Yang, N.~Prasad, and X.~Wang, ``An auction approach to resource allocation
  in uplink OFDMA systems,'' \emph{IEEE Transactions on Signal Processing},
  vol.~57, no.~11, pp. 4482--4496, November 2009.
  
  \item 
M.~Bayati, B.~Prabhakar, D.~Shah, and M.~Sharma, ``Iterative scheduling
  algorithms,'' in \emph{26th IEEE International Conference on Computer
  Communications. (INFOCOM)}, May 2007, pp. 445--453.

\end{enumerate}


\end{document}